\documentclass[a4paper,11pt,oneside]{tSYS2e}

\usepackage[round]{natbib}
\usepackage{amsmath,amsthm,amssymb}
\usepackage{indentfirst}

\usepackage[utf8]{inputenc} 
\usepackage[T1]{fontenc}    
\usepackage{hyperref}       
\usepackage{url}            
\usepackage{booktabs}       
\usepackage{amsfonts}       
\usepackage{nicefrac}       
\usepackage{microtype}      
\usepackage{lipsum}
\usepackage{xcolor}
\usepackage{graphicx,epsfig}
\usepackage{caption}
\usepackage{subcaption}
\usepackage{placeins}
\usepackage{breakcites}
\usepackage{algorithmic}
\usepackage{enumerate}
\usepackage{mathtools}
\usepackage{commath}
\usepackage{color}
\usepackage{epstopdf}
\usepackage{multirow}
\usepackage{scalerel}
\usepackage{tabularx}
\usepackage{float}

\newtheorem{theorem}{Theorem}[section]
\newtheorem{lemma}[theorem]{Lemma}
\newtheorem{proposition}[theorem]{Proposition}
\newtheorem{corollary}[theorem]{Corollary}
\newtheorem{remark}[theorem]{Remark}

\def\cB{\mathcal{B}}

\def\cJ{\mathcal{J}}
\def\cL{\mathcal{L}}
\def\cM{\mathcal{M}}
\def\cN{\mathcal{N}}

\def\R{\mathbb R}
\let\emptyset\varnothing

\makeatletter
\newcommand*{\rom}[1]{\expandafter\@slowromancap\romannumeral #1@}
\makeatother

\jname{Communications in Statistics-Simulation and Computation}

\usepackage{color}

\received{}
\def\draftnote{}

\begin{document}

	\title{Hypothesis testing for populations of networks}
	
	\author{
		\name{Li Chen\textsuperscript{a}, \quad
			Jie Zhou\textsuperscript{b},
			\quad Lizhen Lin\textsuperscript{c}$^{\ast}$\thanks{$^\ast$Corresponding author. E-mail: lizhen.lin@nd.edu}}
		\affil{\textsuperscript{a}College of Mathematics, Southwest Minzu University, Chengdu, Sichuan, China \\
			\textsuperscript{b}College of Mathematics, Sichuan University, Chengdu, Sichuan, China \\
			\textsuperscript{c}Department of Applied and Computational Mathematics and Statistics, \\ University of Notre Dame,
			South Bend, Indiana, USA}
	}

	\maketitle
	
	\begin{abstract}
		It has become an increasingly  common practice in modern science and engineering to collect samples of multiple network data in which a network serves as a basic data object. The increasing prevalence of multiple network data calls for developments of models and theories that can deal with inference problems for populations of networks. 
		In this work, we propose a general procedure for hypothesis testing of networks and in particular, for differentiating  distributions of two samples of networks. We consider a very general framework which allows us to perform test on  large and    sparse networks. Our contribution is two-fold: (1) We propose a test statistics based on the singular value of a generalized Wigner matrix. The asymptotic null distribution of the statistics is shown to follow the Tracy--Widom distribution as the number of nodes tends to infinity.  The test also yields asymptotic  power guarantee with the power tending to one under the alternative; (2) The test procedure is adapted for change-point detection in dynamic networks  which is proven to be consistent in detecting the change-points.  In addition to theoretical guarantees, another appealing feature of this adapted procedure is that it provides a principled and simple method for selecting the threshold that is also allowed to vary with time. Extensive simulation studies and real data analyses demonstrate the superior performance of our procedure with competitors.
	\end{abstract}
	

	\section{Introduction}
	
	One of the  unique features in modern data science is the increasing availability of complex data in non-traditional forms. Among the newer forms of data,  network has arguably emerged as one of the most important and powerful data types.  A network,  an abstract object consisting of a set of nodes and edges, can be broadly used to represent interactions among a set of agents or entities and one can find its applications in virtually any scientific field. The ubiquity of network data in diverse fields ranging from biology \citep{chen2006, cline2007}, physics \citep{Gergana2012Overview, Andrzej2014Modeling}, social science \citep{Hoff2002Latent, Snijders2003A} to engineering \citep{Leonardi2013Tight, chen2010network} has spurred fast developments in models, theories and algorithms for the field of network analysis, see e.g., \cite{1959erdos, holland1983stochastic, newmanblock, ball2011efficient, graphon1,  rohe_yu,  decelle2011,  amini2014semidefinite, bickel2009}. 
	The existing literature, however,  has largely been focusing on  inference of one single (often large) network. The recent advancement in technology and computer prowess has led to the increasing  prevalence  of network data available in multiple networks in which a network serves as the basic data object. For instance, such datasets can be found  in neuroscience \citep{bassett2008}, cancer study  \citep{zhang2009}, microbiome study  \citep{cai2019}, and social interactions  \citep{kossinets2006,eagle2009}. There is a strong need for development of models and theories that can deal with  such data sets, and  more broadly, for inference of population of networks. 
	
	One has already seen a growing effort in this direction. \cite{ginestet2017} proposes a geometric framework for hypothesis tests of populations of networks  viewing a weighted network as a point on a manifold. Along the same line, \cite{paperwitheric} provides  geometric characterization of space of all unlabeled  networks which serve as the foundation for inference based on Fr\'echet mean of networks. In addition, \cite{NIPS2017_7282} provides a general framework for clustering network objects. \cite{2020arXiv200404765J} proposes a Gaussian process based framework for regression and classification with network inputs.  \cite{dunson-durante17} proposes a Bayesian nonparametric approach for modeling the populations of networks.
	
	One of commonly encountered problems for inference of populations of networks is  hypothesis testing which has significant applications,  but remains largely understudied especially for large networks. Among the few existing work in the literature, besides \cite{ginestet2017} as mentioned above, \cite{tang2017}  carries out  hypothesis tests using random dot product graph model via adjacency spectral embedding. \cite{ghoshdastidar2017a} proposes two test statistics based on estimates of the Frobenius norm and spectral norm between link probability matrices of the two samples, the key challenge of which  lies in choosing a  threshold for the test statistics. \cite{ghoshdastidar2018} uses the same statistics as \cite{ghoshdastidar2017a} and proves asymptotic normality for the statistics. \cite{ghoshdastidar2018} further proposes a test statistics based on the extreme eigenvalues of a scaled and centralized matrix and proves that the new statistics asymptotically follows the Tracy--Widom law \citep{tracy1996}. Most of the literature, however, focuses on the case where the number of nodes for each network is fixed, which greatly limits the scope of inference.
	
	The initial focus of our work is on hypothesis testing for two samples of networks including large and sparse networks.  We propose a very intuitive testing statistics which yields theoretical guarantees. More specifically, we prove that its asymptotic null distribution follows the Tracy--Widom distribution and the asymptotic power tends to 1 under the alternative. One of the appealing features of our approach is that our test adopts a very general framework in which the number of the nodes are allowed to grow to infinity, while  most of the existing methods assume that the number of nodes is fixed, which is not always a practical assumption since many modern networks are often large and sparse. We then adapt our test statistics for a change-point detection procedure in dynamic networks and prove its consistency in detecting change-points. We provide a principled method for selecting the threshold level in the change-point detection procedure based on the asymptotic distribution of the testing statistics and the threshold is allowed to vary with time. This is appealing comparing to many existing change-point detection approaches which require either a cross-validation for selecting the threshold or a careful tuning of the parameters. Extensive simulation studies and two real data analyses demonstrate the superior performance of our procedure in comparing with others in both tasks. 
	
	The paper is organized as follows. In Section \ref{sec-twosample}, we propose a testing statistics and throughly study its asymptotic properties. Section \ref{sec-changepoint} is devoted to a change-point detection procedure for dynamic networks by adapting the testing statistics derived in Section \ref{sec-twosample}. Simulation studies are carried out in Section \ref{sec-simu} and real data examples are presented in Section \ref{sec-data}. Technical proofs can be found in the appendix.
	
	\section{Two-sample hypothesis testing for networks}
	\label{sec-twosample}
	
	\subsection{Notation}
	We first introduce some notations that will be used throughout the paper.  For a set $\cN$, $|\cN|$ denotes its cardinality. $TW_1$ denotes the Tracy--Widom distribution with index 1. $\chi^2(n)$ denotes the Chi-squared distribution with $n$ degrees of freedom. 
	For a square matrix $B \in \R^{n \times n}$, $B_{i j}$ denotes its $(i,j)$ entry, $B_{i \cdot}$ is the $i$th row of $B$, and $B_{\cdot i}$ is the $i$th column of $B$. 
	For a symmetric matrix $B \in \R^{n \times n}$, $\lambda_j(B)$ denotes its $j$th largest eigenvalue, ordered as $\lambda_1(B) \geq \lambda_2(B) \geq \dots \geq \lambda_n(B)$, $\sigma_1(B)$ is the largest singular value. Write $X_n \rightsquigarrow X$ if a sequence of random variables $\{X_n\}_{n = 1}^{\infty}$ converges in distribution to random variable $X$. $\lfloor x \rfloor$ denotes the largest integer but no greater than $x \in \R$. $I(\cdot)$ denotes indicator function. For two sequences of real numbers $\{x_n\}$ and $\{y_n\}$, we have the following notations:
	
	\hspace{0.2cm} $y_n = O_n(x_n)$: there exists a positive constant $M$ such that $\underset{n \rightarrow \infty}{\lim} |\frac{y_n}{x_n}| \leq M$.
	
	\hspace{0.2cm} $y_n = o_n(x_n)$: $\underset{n \rightarrow \infty}{\lim}\frac{y_n}{x_n} = 0$.
	
	
	
	\hspace{0.2cm} $y_n = o_p(x_n)$: $\underset{n \rightarrow \infty}{\lim} P\big( \big| \frac{y_n}{x_n} \big| \geq \varepsilon \big) = 0$ for any positive $\varepsilon$.
	
	\subsection{Problem setup and some existing tests} 
	\label{sec:setup}
	We consider two samples of networks with $n$ nodes and sample sizes $m_1$ and $m_2$ respectively.   More specifically, we assume one observes symmetric binary adjacency matrices  $A_1^{(1)},\ldots,A_1^{(m_1)}$ that are generated from symmetric link probability matrix $P_1$ with $A_{1,i j}^{(k)} \sim \text{Bernoulli}(P_{1, i j})$, $k = 1,2,\ldots,m_1$, $i,j = 1,2,\ldots,n$, and another sample of  adjacency matrices $A_2^{(1)},\ldots,A_2^{(m_2)}$  generated from the same model with link probability matrix $P_2$. Our goal is to test whether  the two samples of networks have same graph structure or not, which is equivalent to testing:
	\begin{equation} \label{test for two-sample}
	H_0: P_1 = P_2 ~ \text{against} ~ H_1: P_1 \neq P_2.  
	\end{equation}
	
	For the case of $m_1 = m_2 = 1$ and a fixed $n$, \cite{tang2017} focuses on random dot product graphs by applying the adjacency spectral embedding, whereas \cite{ghoshdastidar2018} focuses on the inhomogeneous Erd\H{o}s--R\'{e}nyi graphs and proposes a test based on eigenvalues. 
	
	For the case of large $m_1, m_2$ and again a fixed number of nodes $n$, \cite{ginestet2017} proposes a $\chi^2$-type test based on a geometric characterization of the space of graph Laplacians and a notion  of Fr\'{e}chet means \citep{frechet, linclt}. As a simplification of the statistics in \cite{ginestet2017}, \cite{ghoshdastidar2018} sets $m_1 = m_2 = m$ and obtains the test statistics as follows:
	\begin{equation}
	T_{\chi^2} = \sum_{i < j} \frac{(\bar{A}_{1,i j} - \bar{A}_{2,i j})^2}{\frac{1}{m (m - 1)} \sum_{k = 1}^m \left(A_{1,i j}^{(k)} - \bar{A}_{1,i j}\right)^2 + \frac{1}{m (m - 1)} \sum_{k = 1}^m \left(A_{2,i j}^{(k)} - \bar{A}_{2,i j}\right)^2}, \label{chi2-type test}
	\end{equation}
	where $\bar{A}_{u,i j} = \frac{1}{m} \sum_{k = 1}^m A^{(k)}_{u,i j}$ with $u = 1, 2$. Then $T_{\chi^2} \rightarrow \chi^2\big(\frac{n (n - 1)}{2}\big)$ as $m \rightarrow \infty$. 
	We call this method $\chi^2$-type test.
	
	The case of large $n$ and fixed $m_1$ and $m_2$ is one of the likely scenarios in practice and  is thus perhaps more interesting. \cite{ghoshdastidar2018} uses the same statistics as \cite{ghoshdastidar2017a} as follows: 
	\begin{equation}
	T_N = \frac{\sum_{i < j} \left( \sum_{k \leq m / 2} A_{1, i j}^{(k)} - A_{2, i j}^{(k)} \right) \left( \sum_{k > m / 2} A_{1, i j}^{(k)} - A_{2, i j}^{(k)} \right)} {\sqrt{ \sum_{i < j} \left( \sum_{k \leq m / 2} A_{1, i j}^{(k)} + A_{2, i j}^{(k)} \right) \left( \sum_{k > m / 2} A_{1, i j}^{(k)} + A_{2, i j}^{(k)} \right) }}. \label{Normal-type test}
	\end{equation}
	\cite{ghoshdastidar2018} proves the asymptotic normality of $T_N$ as $n \rightarrow \infty$. 
	We refer this method to $N$-type test. 
	
	\subsection{Proposed test statistics} 
	\label{subsection General model}
	
	In proposing our test statistics, we consider a very general setting in which the number of nodes can grow to infinity instead of being fixed like in most of the existing literature,  and  the sample sizes $m_1$ and $m_2$  grow in an appropriate rate. 
	We first introduce  the centralized and re-scaled matrix $Z$ with entries given as follows:
	\begin{equation} \label{Z for two-smaple}
	Z_{ij} = \frac{\bar{A}_{1,i j} - \bar{A}_{2,i j}}{\sqrt{ (n - 1) \left[  \frac{1}{m_1}P_{1,i j} \big( 1 - P_{1,i j}\big) + \frac{1}{m_2}P_{2,i j} \big( 1 - P_{2,i j}\big)\right]  }},
	\end{equation}
	where $\bar{A}_{u,i j} = \frac{1}{m_u} \sum_{k = 1}^{m_u} A_{u,i j}^{(k)}$ with $u = 1,2$ and $i, j = 1, \ldots, n$.
	
	The matrix $Z$ involves unknown link probability matrices $P_1$ and $P_2$ thus can not be directly used as a test statistics. As an alternative, one can choose some appropriate  plugin estimates for $P_1$ and $P_2$, and  some of these estimates attain good properties for the resulting tests as we will see in the following discussions.
	
	Denote  $\hat{P}_1$ and $\hat{P}_2$ as some plugin estimators of $P_1$ and $P_2$ respectively, then the empirical standardized matrix $\hat{Z}$ of $Z$ can be written with entries as
	\begin{equation} \label{Z_hat for two-sample}
	\begin{split}	
	\hat{Z}_{ij} = \frac{\bar{A}_{1,i j} - \bar{A}_{2,i j}}{\sqrt{(n - 1) \left[  \frac{1}{m_1} \hat{P}_{1,i j} \big( 1 - \hat{P}_{1,i j}\big) + \frac{1}{m_2} \hat{P}_{2,i j} \big( 1 - \hat{P}_{2,i j}\big) \right] }}, \ 
	i, j = 1, 2, \ldots, n. 
	\end{split}
	\end{equation} 
	
	We  propose to use the largest singular value of $\hat{Z}$, after suitable shifting and scaling, as our test statistics:
	\begin{equation} \label{statistics for two-sample test}
	T_{TW_1} = n^{2/3} \big[ \sigma_1 (\hat{Z}) - 2 \big].
	\end{equation}
	Given a significance level $\alpha \in (0,1)$, the rejection region $Q$ for $H_0$ in test \eqref{test for two-sample} is
	\begin{equation} \label{reject rule for two-sample}
	Q = \{T_{TW_1}|T_{TW_1} \geq \tau_{\alpha / 2}\},
	\end{equation}
	where $\tau_{\alpha / 2}$ is the corresponding $\alpha / 2$ upper quantile of $TW_1$. We then have the following results.
	
	\begin{theorem}[General asymptotic null distribution]
		\label{asymptotic null distribution result for GENERAL}
		Let $A_1^{(1)},\ldots,A_1^{(m_1)}$ be a sample of networks generated from a link probability matrix $P_1$ with $n$ nodes, and $A_2^{(1)},\ldots,A_2^{(m_2)}$ be another sample generated from a link probability matrix $P_2$ with the same number of  nodes. Let $\hat{Z}$ be given as in \eqref{Z_hat for two-sample}. Given some estimated matrices $\hat{P}_u$ of $P_u, u = 1,2$, if $\sup_{i,j} |\hat{P}_{u,i j} - P_{u, ij}| = o_p(n^{- 2 / 3})$, then the following holds under the null hypothesis in \eqref{test for two-sample}:
		\begin{equation} \label{Tw_1 converge}
		n^{2/3} [\lambda_1 (\hat{Z}) - 2] \rightsquigarrow TW_1, ~
		n^{2/3} [ - \lambda_n (\hat{Z}) - 2] \rightsquigarrow TW_1.
		\end{equation}
	\end{theorem}

	\begin{remark}
		Theorem \ref{asymptotic null distribution result for GENERAL} is very general in the sense that it puts no structural conditions on the networks, nor does it impose any assumption on the type of  estimates for  $P_1$ and $P_2$ so long as they are estimated within $o_p(n^{- 2 / 3})$ error. 
	\end{remark}

	The following corollaries show  asymptotic type \rom{1} error control and  asymptotic power  for the rejection rule \eqref{reject rule for two-sample}.

	\begin{corollary} 
		[Asymptotic type \rom{1} error control]
		\label{asymptotic type \rom{1} error for GENERAL}
		Supposing assumptions in Theorem \ref{asymptotic null distribution result for GENERAL} hold, the rejection region in \eqref{reject rule for two-sample} has size $\alpha$.
	\end{corollary}
	
	\begin{corollary}
		[Asymptotic power guarantee]
		\label{asymptotic power guarantee for GENERAL}
		Define a matrix $\tilde{Z} \in \R^{n \times n}$ with zero diagonal and for any $i \neq j$,
		\begin{equation} \label{Z tuta}
		\tilde{Z}_{i j} = \frac{P_{1,i j} - P_{2,i j}}{\sqrt{ (n - 1) \Big[ \frac{1}{m_1}P_{1,i j} \big( 1 - P_{1,i j}\big) + \frac{1}{m_2}P_{2,i j} \big( 1 - P_{2,i j}\big)\Big] }}.
		\end{equation}
		Under the assumptions of Theorem \ref{asymptotic null distribution result for GENERAL}, if $P_{1,i j}$ and $P_{2,i j}$ are such that $n^{- 2 / 3} [ \sigma_1(\tilde{Z}) - 4]^{- 1} \leq o_n(1)$, then
		\begin{equation*}
		P(T_{TW_1} \geq \tau_{\alpha / 2} ) = 1 - o_n(1).
		\end{equation*}
	\end{corollary}

	\begin{remark}
		
		As mentioned  in the introduction, in \cite{ghoshdastidar2018}, a test statistics for comparing two large graphs  is proposed, and our test statistics appears to be similar in natural to theirs.
		However, there are some key distinctions between our method and theirs.  First, our testing statistics considers two-sample test on two populations of networks  which requires exploration of the proper interplay between the asymptotics in both the sample sizes of networks and nodes number.  
		Second, \cite{ghoshdastidar2018} proves the asymptotic Tracy--Widom law under the true link probability matrices, while in our paper, we consider various estimates of  link probability matrices (again based on multiple networks) and  prove the Tracy--Widom law theoretically. We also discuss the performance  of the resulting testing statistics under various estimators. Third, our testing statistics is modified for a novel and efficient change-point detection procedure and  the consistency of the change-point detection is also proved.
		
	\end{remark}

	\subsection{Different estimators of  link probability matrix} 
	\label{subsection Comparision}
	The testing statistics proposed in the previous section requires a plugin estimator for the link probability matrix based on a sample of networks. In this subsection, we investigate the properties of the tests corresponding to various different estimators for link probability matrix.
	
	We first consider a different but natural and simple estimator of $P_u$ by using the average of all the adjacency matrices in the same group. We denote this method as AVG and the link probability matrix estimator as $\hat{P}_{\text{AVG}, u}$, which is actually $\bar{A}_u$.
	
	It's not difficult to see that $$\sup_{i,j}|\hat{P}_{\text{AVG}, u,i j} - P_{u,i j}| = o_p\big( m_u^{- 1 / 2} \log(n)\big)$$ by applying Bernstein's inequality. To guarantee the asymptotic $TW_1$ in \eqref{Tw_1 converge}, it requires that $m_u = O_n(n^{4 / 3})$. 
	More specifically, the sample size $m_u$ needs to increase faster than nodes number $n$, so $m_u$ will exceed $n$ eventually as $n$ tends to infinity. Therefore, the AVG estimator will perform well  if the sample size is large enough. However, this is hard to hold in reality especially when the size of the network is large. Usually, for most practical applications, it would be more suitable to require $m_u$ to increase slower than $n$.
	
	We also consider an average estimator of $P_u$ based on the stochastic block model (SBM), which is similar in spirit to the estimator in \cite{ghoshdastidar2018} but with a different  algorithm for estimating the communities. Our main idea can be summarized as follows: First, assume the graphs are SBMs, or approximate them with SBMs by a weaker version of Szemer\'{e}di's regularity lemma (see \cite{lovasz2012}). Second, use one of the community detection algorithms such as  the goodness-of-fit test proposed in \cite{lei2016} to estimate the number of the communities $\hat{K}_u$. Then perform  clustering using for example the  spectral clustering algorithm (see, e.g.,  \cite{von2007}) to obtain estimates of  the membership vector $g_u \in \{1,\ldots,\hat{K}_u\}^n$ as well as the community set $\cB_{u, k} = \{i: 1 \leq i \leq n, g_{u, i} = k\}$, where $k = 1,2,\ldots,\hat{K}_u$ and $g_{u, i}$ is the $i$th element of $g_u$. Subsequently, $P_u$ is approximated by a block matrix $\hat{P}_{\text{SBM}, u}$ such that $\hat{P}_{\text{SBM}, u, i j}$ is the mean of the submatrix of $\bar{A}_{u}$ restricted to $\cB_{u, g_{u,i}} \times \cB_{u, g_{u,j}}$.
	
	Under further assumption that each community has size at least proportional to $n / K_u$, where $K_u$ is the true community number, it can be seen that the error of $\hat{P}_{\text{SBM}, u, i j} $ is $o_p(K_u m_u^{- 1 / 2} n^{- 1} \log n)$ \citep{lei2016}. This implies that only when $K_u = O_n(n^{\gamma_u}), \gamma_u < 1 / 3 + \alpha_u / 2$, and $m_u = O_n(n^{\alpha_u}), \alpha_u \geq 0$, the error condition in Theorem \ref{asymptotic null distribution result for GENERAL} holds. For large networks in practice,  the number of communities can be very large therefore such a condition might be hard to satisfy. Moreover, due to the potential double estimation in the process  (in estimating the number of communities as well as the  community membership), it may bring large error to the final test statistics, especially when the SBM assumption is not valid.

	We now discuss another explicit method for the link probability matrix estimates that can be used as the plugging estimates in the test statistics called the modified neighborhood smoothing (MNBS) estimator. 
	Let $\{\xi_i\}_{i = 1}^n$ be a 
	random sequence such that $\xi_i, i=1,\dots,n$, are $i.i.d.$ uniform random variables on $[0, 1]$. Conditional on this global sequence $\{\xi_i\}_{i = 1}^n$, we assume all the adjacency matrices $A^{(1)}, A^{(2)},\ldots,A^{(m)}$ in the same population share the same link probability matrix $P \in \R^{n \times n}$, which is modeled by a graphon function $f:[0,1]^2 \rightarrow [0,1]$ such that 
	\begin{align*}
	P_{ij}=f(\xi_i, \xi_j).
	\end{align*}
	Therefore, we have
	\begin{equation*}
	A_{i j}^{(k)} \mid  \{\xi_i\}_{i = 1}^n \sim \text{Bernoulli}(f(\xi_i, \xi_j)), 
	\end{equation*}
	independently for all $i \leq j$ and $k = 1, \dots, m$.
	
	We then 
	apply MNBS method proposed in \cite{zhao2019} to estimate $P$. 
	The essential idea of the MNBS procedure consists of the following steps: First, for the group of adjacency matrices $A^{(1)}, A^{(2)},\ldots,A^{(m)}$ generated from $P$, let $\bar{A} = \sum_{k = 1}^m A^{(k)} / m$, define the distance measure between nodes $i$ and $i'$ as $d^2(i, i') = \max_{k \neq i, i'} |\langle\bar{A}_{i \cdot} - \bar{A}_{i' \cdot},\bar{A}_{k \cdot}\rangle|$ and the neighborhood of node $i$ as $\cN_i = \{i' \neq i: d^2(i, i') \leq q_i(q)\}$, where $q_i(q)$ denotes the $q$th quantile of the distance set $\{d^2(i, i'): i' \neq i\}$. Then the parameter $q$ is set to be $C \log n / (n^{1 / 2} \omega)$, where $C$ is some positive constant and $\omega = \min\{ n^{1 / 2}, (m \log n)^{1 / 2}\}$. Finally, given the neighborhood $\cN_i$ for each node $i$, the link probability $P_{i j}$ between nodes $i$ and $j$ is estimated by $\tilde{P}_{ij} = \sum_{i' \in \cN_i} \bar{A}_{i' j} / |\cN_i|.$ In comparing with the neighborhood smoothing method proposed in \cite{zhang2017}, the key idea is to employ the average network information $\Bar A$ and simultaneously  shrink  the neighborhood size (from $C( \log n / n)^{1 / 2} $ to $C \log n / (n^{1 / 2} \omega)$) to obtain an estimate with an improved rate. 
	
	Based on MNBS, for the symmetric networks considered in this paper, we use symmetrized estimators of the link probability matrices $P_u, u = 1, 2$, of the two groups of graphs as
	\begin{align} \label{P estimated by MNBS}
	\hat{P}_u = \frac{\tilde{P}_u + (\tilde{P}_u)^T}{2}, ~ \text{with } \tilde{P}_{u,ij} = \frac{\sum_{i' \in \cN_{u,i}} \bar{A}_{u,i' j}}{|\cN_{u,i}|},
	\end{align}
	where $\bar{A}_{u,i' j}$ is the $(i', j)$ element of $\bar{A}_u = \sum_{k = 1}^{m_u} A_u^{(k)} / m_u$ and $\cN_{u,i}$ is the neighborhood of node $i$ in group $u$.
	
	From Lemma 9.3 in \cite{zhao2019}, we have 
	\begin{equation} \label{neighbor number}
	|\cN_{u,i}| \geq B_u \frac{n^{1 / 2} \log n}{\omega_u},
	\end{equation}
	where $B_u$ is a global positive constant and $\omega_u = \min\{ n^{1 / 2}, (m_u \log n)^{1 / 2}\}$ for $u=1,2$.
	
	For the MNBS, we do not provide an explicit rate on bounding the sup norm $\sup_{i,j} |\hat{P}_{u,i j} - P_{u, i j}| $ due to the difficulty in deriving the point-wise rate. 
	From the definition of $d^2(i, i’)$, it can be seen that the distance measure between nodes $i$ and $i’$ in the MNBS algorithm is based on the row pattern similarity instead of point-wise way. To derive an entry-wise error of $P_{ij}$, one can use Bernstein's inequality, but the neighbor of a node is selected by the $q$th quantile of the distance set, which would decrease the variance of sample in the neighbor, but this decreased variance is unknown.
	The extensive simulation carried out in Section \ref{sec-simu} and Section \ref{sec-data} show that MNBS-based tests often yield the best performance in comparing with tests based on other estimators.

	\begin{remark}
		As one can see in our setup in Section \ref{sec:setup}, it is assumed that the edges of each network $A^{(k)}$, $k=1,\ldots, m$, in the same populations are generated \emph{independently} from the same deterministic link probability matrix $P$. To fit this setup under a  genuine `graphon model', one has to assume that for each node $i$, the latent variable $\xi_i$ is the same over all the samples in the same population and will not change for each sample. That is, a global latent sequence $\{\xi_i\}_{i=1}^n$ is shared across all the networks. Note that this setup does not fall under a genuine graphon model in which one first samples uniform random sequence $\{\xi^{(k)}_i\}_{i=1}^n$ for each network over $k$, then 
		$A_{ij}^{(k)}\mid \{\xi^{(k)}_i\}_{i=1}^n \sim \text{Bernoulli}(f(\xi^{(k)}_i, \xi^{(k)}_j))$. Therefore, the entries of the adjacent matrix or network are not independent after marginalizing the latent variables. 
		
	\end{remark}

	\section{Change-point detection in dynamic networks}
	\label{sec-changepoint}
	We refer the two sample test based on asymptotic $TW_1$ proposed in the previous section as $TW_1$-type test. In this section, we adapt the $TW_1$-type test to  a procedure for  change-point detection in dynamic networks, which is another important learning task in statistics and  has received a great deal of recent attentions. Specifically, we examine a sequence of networks whose distributions may exhibit changes at some time epochs. Then, the problem is to determine the unknown change-points based on the observed sequence of network adjacency matrices. 
	
	Assume the observed dynamic networks $\{A_t\}_{t = 1}^m$ are generated by a sequence of probability matrices $\{P_t\}_{t = 1}^m$ with $A_{t, i j} \sim \text{Bernoulli}(P_{t, i j})$ for time $t = 1,\ldots, m$. Let $\cJ = \{\eta_j\}_{j = 1}^{J} \subset \{1,\ldots,m\}$ be a collection of change-points and $\eta_0 = 0$, $\eta_{J + 1} = m$,  ordered as $\eta_0 < \eta_1 < \dots < \eta_J < \eta_{J + 1}$, such that
	\begin{equation*}
	P_t = P^{(j)}, t = \eta_{j - 1} + 1, \ldots, \eta_j, j = 1, \ldots, J + 1.
	\end{equation*}
	In other words, the change-points $\{\eta_j\}_{j = 1}^{J}$ divide the networks into $J + 1$ groups, the networks contained in the same group follow the same link probability matrix and $P^{(j)}$ is the link probability matrix of the $j$th segment satisfying $P^{(j)} \neq P^{(j + 1)}$. Denote $\cJ = \emptyset$ if $J = 0$.
	
	Now we apply our $TW_1$-type test to a screening and thresholding algorithm that is commonly used in change-point detection, see \cite{niu2012, zou2014, zhao2019}. The detection procedure is referred as $TW_1$-type detection and  described as follows.
	
	Define $L = \min_{1 \leq j \leq J + 1}(\eta_j - \eta_{j - 1})$, which is the minimum segment length. Set a screening window size $h \ll m$ and $h < L / 2$. Denote $\bar{A}_1(t, h) = \frac{1}{h} \sum_{i = t - h + 1}^{t} A_i$ and $\bar{A}_2(t, h) = \frac{1}{h} \sum_{i = t + 1}^{t + h} A_i$ for each $t = h, \ldots, m - h$. $\hat{P}_1(t, h)$ and $\hat{P}_2(t, h)$ are  for example MNBS estimators using $\{A_i\}_{i = t - h + 1}^{t}$ and $\{A_i\}_{i = t + 1}^{t + h}$ respectively. In addition, we denote a matrix $\hat{Z}(t,h)$ with entries as follows essentially the same as in \eqref{Z_hat for two-sample}:
	
	\begin{equation*}
	\begin{split}
	\hat{Z}_{i j}(t, h) = \frac{\bar{A}_{1, i j}(t, h) -\bar{A}_{2, i j}(t, h)}{\sqrt{(n - 1) \left\{ \frac{1}{h} \hat{P}_{1, i j}(t, h) \Big[ 1 - \hat{P}_{1, i j}(t, h)\Big]  + \frac{1}{h} \hat{P}_{2, i j}(t, h) \Big[ 1 - \hat{P}_{2, i j}(t, h)\Big] \right\}  }}, \\
	i, j = 1, 2, \ldots, n.
	\end{split}
	\end{equation*}	
	In the screening step, we calculate the scan statistics $T_{TW_1} (t, h)$ depending only on observations in a small neighborhood $[t - h + 1, t + h]$ as follows:
	\begin{equation*}
	T_{TW_1}(t,h) = n^{2 / 3} \big\{\sigma_1\big[\hat{Z}(t,h)\big] - 2\big\}.
	\end{equation*}
	Define the $h$-local maximizers of $T_{TW_1}(t,h)$ as $\{t: T_{TW_1}(t,h) \geq T_{TW_1}(t',h) ~ \text{for all} ~ t' \in (t - h, t + h)\}$. Let $\cL \cM$ denote the set of all $h$-local maximizers of $T_{TW_1} (t, h)$. 
	
	In the thresholding step, we estimate the change-points by a thresholding rule to $\cL \cM$ with time $t$ such that
	\begin{equation} \label{CP time}
	\hat{\cJ} = \{t: t \in \cL \cM ~ \text{and} ~ T_{TW_1} (t, h) > \vartriangle_{T_{TW_1}}\},
	\end{equation}
	where $\vartriangle_{T_{TW_1}} = \max\{\tau_{\alpha}, n^{2 / 3} [\delta(t, h) - 4] - \tau_{\alpha}\}, \alpha = 1 / 2 - (1 - 1 / n)^{1 / (2 h)} / 2, \delta(t, h) = \sigma_1(V_1(t,h))$ is the largest singular value of matrix $V_1(t,h)$ with zero diagonal and for any $i \neq j$,
	\begin{equation*}
	\begin{split}
	V_{1, i j} (t,h) = \frac{\hat{P}_{1, i j}(t, h) - \hat{P}_{2, i j}(t, h)}{\sqrt{ (n - 1)\Big\{\frac{1}{h} \hat{P}_{1, i j}(t, h) \Big[ 1 - \hat{P}_{1, i j}(t, h)\Big] + \frac{1}{h} \hat{P}_{2, i j}(t, h) \Big[ 1 - \hat{P}_{2, i j}(t, h)\Big]\Big\} }},\\
	i, j = 1, 2, \ldots, n.
	\end{split}
	\end{equation*}
	We have the following consistency result.
	
	\begin{theorem}[Consistency of $TW_1$-type change-point detection]
		\label{consistency of change-points detection}
		Under the alternative hypothesis, assume $n^{2 / 3} [\sigma(t,h) - 4] \geq 2 \tau_{\alpha}, \alpha = 1 / 2 - (1 - 1 / n)^{1 / (2 h)} / 2, h < L / 2$, then the $TW_1$-type change-point detection procedure satisfies
		$$\lim_{n \rightarrow \infty} P\big(\cJ = \hat{\cJ} \big) = 1.$$
	\end{theorem}
	
	One of the interesting findings from Theorem \ref{consistency of change-points detection} is that for a fixed window size $h$, the threshold in \eqref{CP time} is dynamic with time $t$ instead of being a constant as in \cite{zhao2019}. By adapting the $TW_1$-type test for change-point detection, we can adjust the threshold with $t$ and still enjoy consistency of the change-point detection. 
	From the proof of Theorem \ref{consistency of change-points detection}, it is reflected that for a time $t$ that does not correspond to a change-point, $T_{TW_1} (t, h) \leq \vartriangle_{T_{TW_1}}$ with probability 1, so it can control the type \rom{1} error. However, for a change-point $t$, $T_{TW_1} (t, h) > \vartriangle_{T_{TW_1}}$ with probability 1, and hence the  threshold can lead to a good performance.
	
	The only tuning parameter of $TW_1$-type change-point detection procedure is the local window size $h$, which is chosen according to applications with available information or artificially like set $h = \sqrt{m}$ as recommended in \cite{zhao2019}.

	\section{Simulation study}
	\label{sec-simu}
	In this section, we illustrate the performance of $TW_1$-type test and its application to change-point detection using several synthetic data examples. 
	
	We first define four graphons and an SBM, which are used for  two-sample test and change-point detection in the  simulation studies. The graphons are partly borrowed from \cite{zhang2017} and the SBM is from \cite{zhao2019} with 2 communities. We denote the block matrix or the probability matrix of connections between blocks as $\Lambda$. More specifically, the graphons and SBM are defined as:
	
	\hspace{0.5cm} \textbf{Graphon 1}:
	$$f(u,v) = 
	\begin{cases}
	k / (K + 1), & u,v \in ((k - 1) / K, k / K ), \\
	0.3 / (K + 1), & \text{otherwise},
	\end{cases}$$
	where $K = \lfloor \log n \rfloor, k = 1, 2, \ldots, K$.
	
	\hspace{0.5cm} \textbf{Graphon 2}: 
	\begin{align*}f(u,v) = (u^2 + v^2) / 3 \cos [1 / (u^2 + v^2)] + 0.15.
	\end{align*}
	
	\hspace{0.5cm} \textbf{Graphon 3}: 
	\begin{align*} f(u,v) = \sin[5 \pi (u + v -1) + 1]/2 + 0.5.
	\end{align*}
	
	\hspace{0.5cm} \textbf{Graphon 4}: 
	\begin{align*} f(u,v) = (u^2 + v^2)/10 \cos[1/(u^2 + v^2)] + 0.05.
	\end{align*}
	
	\hspace{0.5cm} \textbf{SBM 1}: 
	\begin{align*}\Lambda = 
	\left[
	\begin{matrix}
	0.6 + \theta_0 & 0.3 \\
	0.3 & 0.6
	\end{matrix}
	\right],
	\end{align*}
	where $\theta_0$ is a constant related to sample size $m$. The membership of the $i$th node is $M(i) = I(1 \leq i \leq \lfloor 2 n / \log n \rfloor) + 2 I(\lfloor 2 n / \log n \rfloor + 1 \leq i \leq n)$.
	
	To operationalize simulations related to MNBS, the quantile parameter $q = B_0 (\log n)^{1 / 2} / (n^{1 / 2} h^{1 / 2})$ and the threshold $\vartriangle_D = D_0 (\log n)^{1 / 2 + \delta_0} / (n^{1 / 2} h^{1 / 2})$ with tuning parameters $D_0$ and $\delta_0$ for change-point detection in \cite{zhao2019} need to be specified. In the following simulations in this section and the real data analyses in Section \ref{sec-data}, we set the related parameters $h = \sqrt{m}, B_0 = 3, \delta_0 = 0.1, D_0 = 0.25$ as recommended in \cite{zhao2019} unless otherwise indicated.
	
	\subsection{Two-sample test with simulated data}
	To examine the performance of the two-sample test \eqref{test for two-sample}, we present our results by $TW_1$-type tests based on MNBS ($TW_1$-MNBS), AVG ($TW_1$-AVG), and SBM ($TW_1$-SBM) discussed in subsection  \ref{subsection Comparision}, $\chi^2$-type test with statistics \eqref{chi2-type test}, and $N$-type test with statistics \eqref{Normal-type test}. We measure the performance in terms of the Attained Significance Level (ASL) which is the probability of observing a statistics far away from the true value under the null hypothesis, and the Attained Power (AP), the probability of correctly  rejecting the null hypothesis when the alternative hypothesis is true.  
	
	We conduct two experiments  using Graphon 1 and Graphon 2 respectively. In the first experiment, we generate two groups of networks $\{A_1^{(k)}\}_{k= 1}^{m_1}$ and $\{A_2^{(k)}\}_{k= 1}^{m_2}$.  We vary the number of nodes  $n$ growing from 100 to 1000 in a step of 100 with sample sizes $m_1 = m_2 = 30, 200$, and set significance level at $\alpha = 0.05$. $\{A_1^{(k)}\}_{k= 1}^{m_1}$ are generated from Graphon 1.
	Under the null hypothesis, $\{A_2^{(k)}\}_{k= 1}^{m_2}$ are also generated from the Graphon 1 and hence $P_1 = P_2$. Under the alternative hypothesis, randomly choose $\lfloor \log n \rfloor$-element subset $S \subset \{1, 2, \ldots, n\}$, generate $\{A_2^{(k)}\}_{k= 1}^{m_2}$ from $P_2$ by setting $P_{2, i j} = P_{1, i j} + \theta_1$ with $\theta_1 = 0.05$ for $m_1 = m_2 = 30 ~ (\theta_1 = 0.02 \text{ for } m_1 = m_2 = 200)$ if $i, j \in S$, and $\theta_1 = 0$ otherwise. Using $TW_1$-MNBS, $TW_1$-AVG, $TW_1$-SBM tests, $\chi^2$-type test and $N$-type test, we run 1000 Monte Carlo simulations for the experiment to estimate the ASLs and APs of test \eqref{test for two-sample}. 
	
	The second experiment is conducted similarly but using Graphon 2. The only difference is that 
	for a better visualization of comparisons, under the alternative hypothesis, we set $P_{2, i j} = P_{1, i j} + \theta_2$ with $\theta_2 = 0.2$ for $m_1 = m_2 = 30 ~ (\theta_2 = 0.17 \text{ for } m_1 = m_2 = 200)$ if $i, j \in S$ and $\theta_2 = 0$ otherwise. 
	The rates of rejecting the null hypothesis for these two experiments are summarized in Figures \ref{test simulation 1} and \ref{test simulation 2} respectively.
	
	\begin{figure}[htbp] 
		\centering
		\includegraphics[width = .65\textwidth]{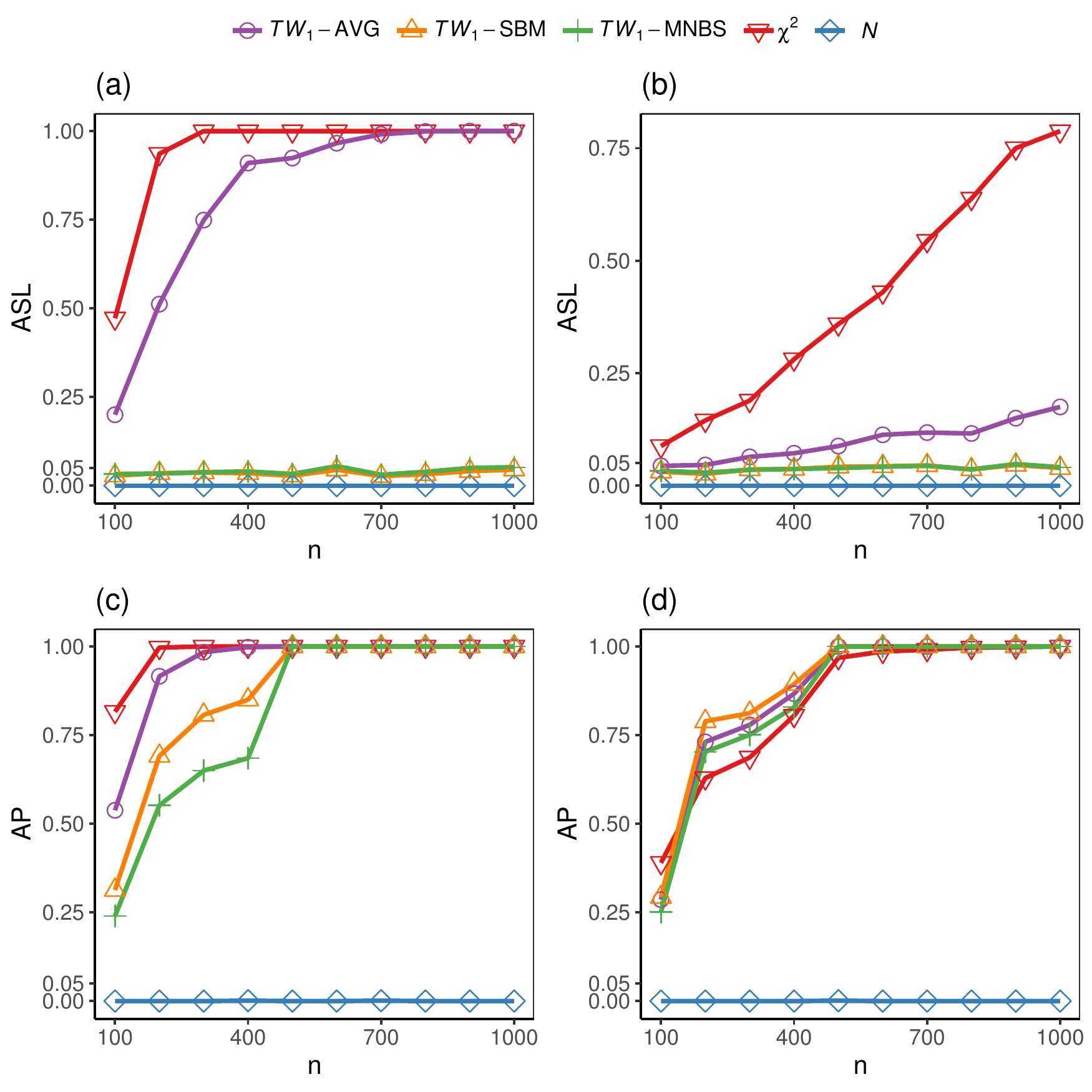} 
		\caption{ASLs and APs of tests using Graphon $1$ for different values of nodes number $n$, sample sizes $m_1$ and $m_2$. $m_1 = m_2 = 30$ for (a) and (c) and $m_1 = m_2 = 200$ for (b) and (d).} \label{test simulation 1}
	\end{figure}
	
	The results of the first experiment using Graphon 1, an SBM set up, are plotted in Figure \ref{test simulation 1}. It reveals  undesirable behaviors of $\chi^2$-type test and $TW_1$-AVG test since with increasing number of nodes $n$, the ASLs of both tests grow quickly close to $1$, which is too large to be used in practice. We can also see that the $N$-type test is not efficient  as both ASLs and APs of the test are $0$ for both cases of $m_1 = m_2 = 30, 200$. Its poor performance in APs is partly due to the small difference between $\{A_1^{(k)}\}_{k = 1}^{m_1}$ and $\{A_2^{(k)}\}_{k = 2}^{m_2}$ we set. However, the performance of $TW_1$-SBM test and $TW_1$-MNBS test are much better, ASLs of both tests are stable and close to the significance level of $\alpha = 0.05$, while APs improve to $1$ as $n$ grows. It is also found that when $n$ is not that large, $TW_1$-SBM test is slightly more powerful in terms of AP than $TW_1$-MNBS test. This is not surprising because the networks generated from Graphon $1$ are endowed with an SBM structure. 
	
	\begin{figure}[htbp] 
		\centering
		\includegraphics[width = .65\textwidth]{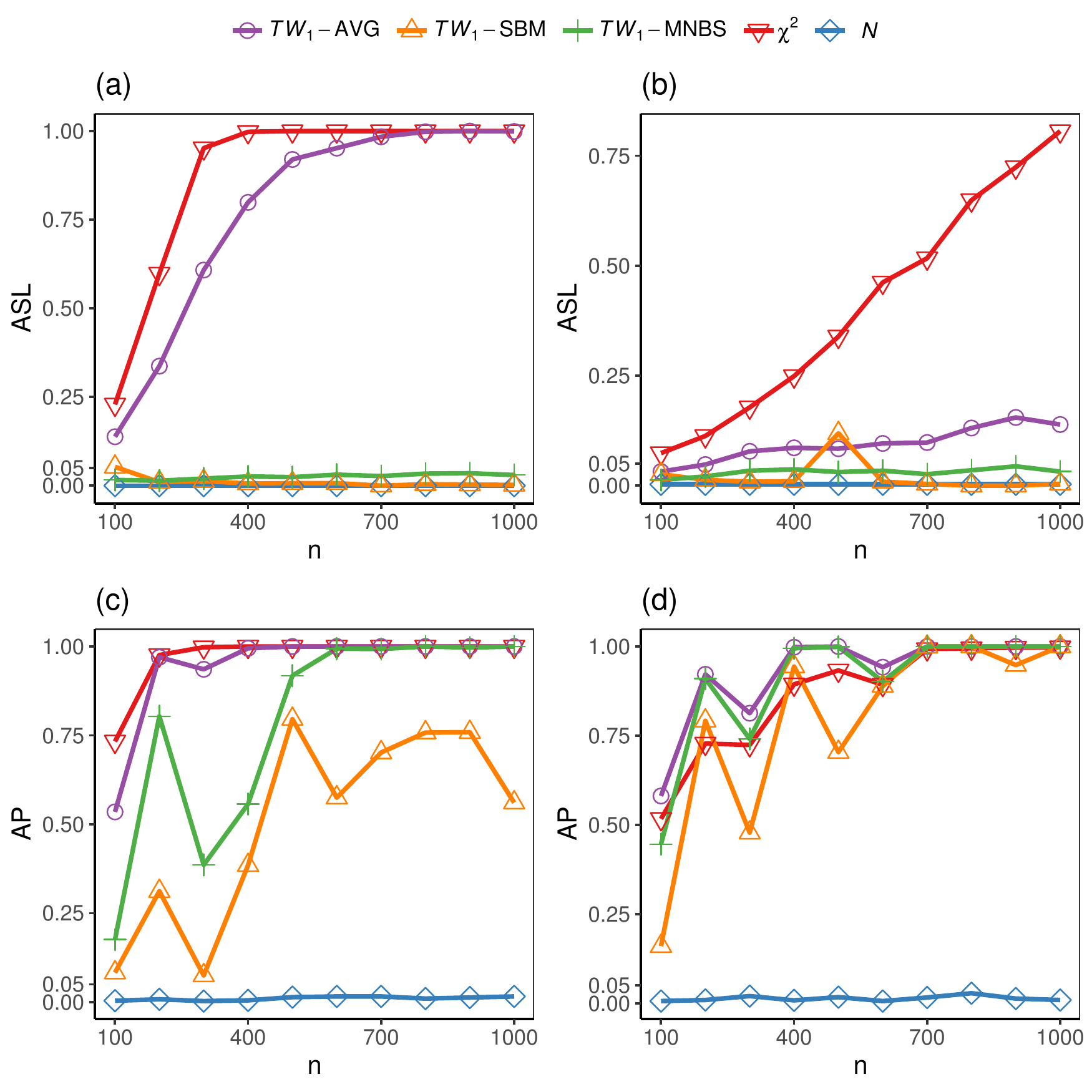} 
		\caption{ASLs and APs of tests using Graphon $2$ for different values of nodes number $n$, sample sizes $m_1$ and $m_2$. $m_1 = m_2 = 30$ for (a) and (c) and  $m_1 = m_2 = 200$ for (b) and (d).} \label{test simulation 2}
	\end{figure}
	
	The results of the second experiment using Graphon 2, which is not an SBM, are given in Figure \ref{test simulation 2}. It indicates that the behaviors of $TW_1$-AVG test, $\chi^2$-type test and $N$-type test are similar to those in the first experiment using Graphon 1 and the performance is poor. On the other hand, $TW_1$-MNBS test has a superior performance than $TW_1$-SBM in both ASL and AP. Specifically, ASLs of $TW_1$-SBM test are away from $0.05$, whereas $TW_1$-MNBS test still performs well on both ASL and AP. Moreover, this also indicates that $TW_1$-SBM test is sensitive to the network structure especially deviation from an SBM. Hence, $TW_1$-MNBS test is more robust to the network structure whereas $TW_1$-SBM test is preferable for SBM networks.
	
	\subsection{Change-point detection in dynamic networks}
	
	To assess the performance of $TW_1$-type change-point detection in dynamic networks, we compare its performance based on MNBS, AVG, and SBM estimators (referred as CP-TWMNBS, CP-TWAVG, CP-TWSBM respectively) to the graph-based nonparametric testing procedure in \cite{chen2015} referred as CP-GRA detection, and the MNBS-based change-point detect procedure in \cite{zhao2019} referred as CP-DMNBS detection. 
	
	Specifically, using all the above five methods, we conduct change-point detection experiments  under three different scenarios with zero, one, and three change-points respectively. For all the experiments, we vary the nodes number and the sample size at $n = 100, 200, 300, m = 100, 200$, and set the significance $\alpha = 0.05$. For each combination of the sample size, nodes number, and the network model, we run 100 Monte Carlo trials. Simultaneously, we also  explore the effect of network sparsity on the performance of change-point detection. For this, we consider the above setting, but scale the link probability $P$ as $\rho P$ by a factor $\rho = {1, 0.25}$, where $\rho = 1$ is exactly the same as the above setting while $\rho = 0.25$ corresponds to sparser graphs. 
	
	\subsubsection{No change-point detection}
	To study the performance with respect to false positives, we simulate two kinds of dynamic networks $\{A_t\}_{t = 1}^m$ with no change-point from Graphon 3 and SMB 1 with $\theta_0 = 0$ respectively. Tables \ref{tab 0-1} and \ref{tab 0-2} report the average number of estimated change-points by using the five methods .
	
	As one can see, the performance of CP-TWSBM, CP-TWMNBS, and CP-GRA detections perform reasonably well and improves as $n$ increases. CP-TWAVG detect method performs well in the  case of Graphon 3 while experiences heavy inflated levels in the  case of SBM 1. As for CP-DMNBS detection, the empirical type \rom{1} error is completely controlled at the target level 0.05 for  SBM 1, but there are some false positives in the case of Graphon 3.
	
	\begin{table}[H]
		\centering
		\caption{Average estimated change-points number $\hat{J}$ under no change-point scenarios through Graphon 3.}
		{\scalebox{0.75}{
				\begin{tabular}{*{8}{c}}
					\toprule			
					$m$ & $n$ & $\rho$ & CP-TWAVG & CP-TWSBM & CP-TWMNBS & CP-GRA & CP-DMNBS \\
					\midrule
					\multirow{2}{*}{100} & \multirow{2}{*}{100} & 1 &
					0.00	&	0.00	&	0.00	&	0.04	&	3.75	\\								
					& & $0.25$ &
					0.00	&	0.00	&	0.00	&	0.03	&	2.04	\\						
					\midrule
					\multirow{2}{*}{100} & \multirow{2}{*}{200} & 1 &
					0.00	&	0.00	&	0.00	&	0.11	&	0.15	\\							
					& & $0.25$ &
					0.00	&	0.00	&	0.00	&	0.07	&	0.3	\\		    			
					\midrule
					\multirow{2}{*}{100} & \multirow{2}{*}{300} & 1 &
					0.00	&	0.00	&	0.00	&	0.04	&	0.02	\\				
					& & $0.25$ &
					0.00	&	0.00	&	0.00	&	0.04	&	0.02	\\				
					\midrule
					\multirow{2}{*}{200} & \multirow{2}{*}{100} & 1 &
					0.00	&	0.00	&	0.00	&	0.02	&	2.02	\\				
					& & $0.25$ &
					0.00	&	0.00	&	0.00	&	0.08	&	5.16	\\				
					\midrule
					\multirow{2}{*}{200} & \multirow{2}{*}{200} & 1 &
					0.00	&	0.00	&	0.00	&	0.08	&	0.21	\\				
					& & $0.25$ &
					0.00	&	0.00	&	0.00	&	0.03	&	0.35	\\				
					\midrule
					\multirow{2}{*}{200} & \multirow{2}{*}{300} & 1 &
					0.00	&	0.00	&	0.00	&	0.05	&	0.01	\\				
					& & $0.25$ &
					0.00	&	0.00	&	0.00	&	0.08	&	0.01	\\				
					\bottomrule
		\end{tabular}} }\label{tab 0-1}
	\end{table}

	\begin{table}[htbp]
		\centering
		\caption{Average estimated change-points number $\hat{J}$ under no change-point scenarios through SBM 1.}
		{\scalebox{0.75}{
				\begin{tabular}{*{8}{c}}
					\toprule			
					$m$ & $n$ & $\rho$ & CP-TWAVG & CP-TWSBM & CP-TWMNBS & CP-GRA & CP-DMNBS \\
					\midrule
					\multirow{2}{*}{100} & \multirow{2}{*}{100} & 1 &
					4.19	&	0.02	&	0.02	&	0.10	&	0.00	\\				
					& & $0.25$ &
					0.06	&	0.02	&	0.02	&	0.02	&	0.00	\\				
					\midrule
					\multirow{2}{*}{100} & \multirow{2}{*}{200} & 1 &
					4.21	&	0.00	&	0.00	&	0.04	&	0.00	\\				
					& & $0.25$ &
					0.16	&	0.00	&	0.03	&	0.03	&	0.00	\\				
					\midrule
					\multirow{2}{*}{100} & \multirow{2}{*}{300} & 1 &
					4.09	&	0.01	&	0.01	&	0.04	&	0.00	\\				
					& & $0.25$ &
					0.17	&	0.02	&	0.03	&	0.01	&	0.00	\\				
					\midrule
					\multirow{2}{*}{200} & \multirow{2}{*}{100} & 1 &
					5.74	&	0.02	&	0.02	&	0.08	&	0.00	\\				
					& & $0.25$ &
					0.42	&	0.02	&	0.07	&	0.04	&	0.00	\\				
					\midrule
					\multirow{2}{*}{200} & \multirow{2}{*}{200} & 1 &
					6.34	&	0.03	&	0.03	&	0.07	&	0.00	\\				
					& & $0.25$ &
					1.13	&	0.02	&	0.02	&	0.01	&	0.00	\\				
					\midrule
					\multirow{2}{*}{200} & \multirow{2}{*}{300} & 1 &
					6.26	&	0.02	&	0.02	&	0.03	&	0.00	\\				
					& & $0.25$ &
					2.30	&	0.01	&	0.01	&	0.02	&	0.00	\\				
					\bottomrule
		\end{tabular}} }\label{tab 0-2}
	\end{table}

	\subsubsection{Single change-point detection}
	We now assess the accuracy of our proposed $TW_1$-type change-point estimators in different scenarios. The dynamic networks $\{A_t\}_{t = 1}^m$ are designed as follows. For $t = 1, 2, \ldots, m / 2$, $A_t$ is generated from link probability matrix $P_1$ by SBM 1 with $\theta_0 = 0$. For $t = m / 2 + 1, \ldots, m$, $A_t$ is generated from $P_2$ by SBM 1 with $\theta_0 = - m^{- 1/4}$. 
	
	We adopt Boysen distance suggested in \cite{boysen2009} as a measurement in the change-point estimation. Specifically, calculate the distances between the estimated change-point set $\hat{\cJ}$ and the true change-point set $\cJ$ as $\varepsilon(\hat{\cJ} \| \cJ) = \max_{b \in \cJ} \min_{a \in \hat{\cJ}} |a - b|$ and $\varepsilon(\cJ \|\hat{\cJ}) = \max_{b \in \hat{\cJ}} \min_{a \in \cJ} |a - b|$.
	
	Utilizing CP-TWMNBS, CP-TWAVG, CP-TWSBM, CP-GRA, and CP-DMNBS detections, we estimate the efficient detect rate (the rate at least one change-point is detected over 100 simulations), the average change-point number over the efficient detections, and the average Boysen distances over the efficient detections. The corresponding results are listed in Tables \ref{tab 1-1}--\ref{tab 1-3}. 
	
	\begin{table}[p]
		\centering
		\caption{Average estimated change-points number $\hat{J}$ under single change-point scenarios through SBM 1.}
		{\scalebox{0.75}{
				\begin{tabular}{*{8}{c}}
					\toprule		
					$m$ & $n$ & $\rho$ & CP-TWAVG & CP-TWSBM & CP-TWMNBS & CP-GRA & CP-DMNBS \\
					\midrule
					\multirow{2}{*}{100} & \multirow{2}{*}{100} & 1 &				
					3.43	&	1.01	&	1.01	&	1.07	&	1.00	\\								
					& & $0.25$ &
					1.02	&	0.95	&	0.98	&	0.00	&	0.00	\\				
					\midrule
					\multirow{2}{*}{100} & \multirow{2}{*}{200} & 1 &
					3.50	&	1.01	&	1.01	&	1.09	&	1.00	\\								
					& & $0.25$ &
					1.04	&	1.00	&	1.01	&	0.00	&	0.00	\\				
					\midrule
					\multirow{2}{*}{100} & \multirow{2}{*}{300} & 1 &
					3.49	&	1.01	&	1.01	&	1.11	&	1.00	\\								
					& & $0.25$ &
					1.04	&	1.00	&	1.01	&	0.00	&	0.00	\\				
					\midrule
					\multirow{2}{*}{200} & \multirow{2}{*}{100} & 1 &
					5.09	&	1.00	&	1.00	&	1.06	&	1.00	\\				
					& & $0.25$ &
					1.05	&	0.66	&	0.73	&	0.04	&	0.00	\\				
					\midrule
					\multirow{2}{*}{200} & \multirow{2}{*}{200} & 1 &
					5.52	&	1.03	&	1.03	&	1.11	&	1.00	\\				
					& & $0.25$ &
					1.51	&	1.00	&	1.01	&	0.00	&	0.00	\\					
					\midrule
					\multirow{2}{*}{200} & \multirow{2}{*}{300} & 1 &
					5.34	&	1.00	&	1.01	&	1.06	&	1.00	\\				
					& & $0.25$ &
					2.19	&	1.00	&	1.00	&	0.00	&	0.00	\\				
					\bottomrule
		\end{tabular}} }\label{tab 1-1}
	\end{table}
	
	\begin{table}[htbp]
		\centering
		\caption{Average efficient detect rate under single change-point scenarios through SBM 1.}
		{\scalebox{0.75}{
				\begin{tabular}{*{8}{c}}
					\toprule			
					$m$ & $n$ & $\rho$ & CP-TWAVG & CP-TWSBM & CP-TWMNBS & CP-GRA & CP-DMNBS \\
					\midrule	
					\multirow{2}{*}{100} & \multirow{2}{*}{100} & 1 &
					1.00	&	1.00	&	1.00	&	1.00	&	1.00	\\				
					& & $0.25$ &
					0.97	&	0.94	&	0.95	&	0.00	&	0.00	\\				
					\midrule
					\multirow{2}{*}{100} & \multirow{2}{*}{200} & 1 &
					1.00	&	1.00	&	1.00	&	1.00	&	1.00	\\				
					& & $0.25$ &
					1.00	&	1.00	&	1.00	&	0.00	&	0.00	\\				
					\midrule
					\multirow{2}{*}{100} & \multirow{2}{*}{300} & 1 &
					1.00	&	1.00	&	1.00	&	1.00	&	1.00	\\				
					& & $0.25$ &
					1.00	&	1.00	&	1.00	&	0.00	&	0.00	\\				
					\midrule
					\multirow{2}{*}{200} & \multirow{2}{*}{100} & 1 &
					1.00	&	1.00	&	1.00	&	1.00	&	1.00	\\				
					& & $0.25$ &
					0.88	&	0.65	&	0.70	&	0.03	&	0.00	\\				
					\midrule
					\multirow{2}{*}{200} & \multirow{2}{*}{200} & 1 &
					1.00	&	1.00	&	1.00	&	1.00	&	1.00	\\				
					& & $0.25$ &
					1.00	&	1.00	&	1.00	&	0.00	&	0.00	\\				
					\midrule
					\multirow{2}{*}{200} & \multirow{2}{*}{300} & 1 &
					1.00	&	1.00	&	1.00	&	1.00	&	1.00	\\				
					& & $0.25$ &
					1.00	&	1.00	&	1.00	&	0.00	&	0.00	\\				
					\bottomrule
		\end{tabular}}} \label{tab 1-2}
	\end{table}
	
	\begin{table}[htbp]
		\centering
		\caption{Average Boysen distances $\varepsilon_1, \varepsilon_2$ under single change-point scenarios through SBM 1.}
		{\scalebox{0.75}{
				\begin{tabular}{*{9}{c}}
					\toprule		
					$m$ & $n$ & $\rho$ & & CP-TWAVG & CP-TWSBM & CP-TWMNBS & CP-GRA & CP-DMNBS \\
					\midrule
					\multirow{4}{*}{100} & \multirow{4}{*}{100} & \multirow{2}{*}{1} & $\varepsilon_1$ &
					35.16	&	0.39	&	0.39	&	1.53	&	0.03	\\				
					& & & $\varepsilon_2$ &
					0.00	&	0.00	&	0.00	&	0.00	&	0.03	\\							
					\cmidrule{3-9}
					& & \multirow{2}{*}{0.25} & $\varepsilon_1$ &
					1.62	&	0.39	&	0.92	&	-	&	-	\\								
					& & & $\varepsilon_2$ &
					0.13	&	0.12	&	0.12	&	-	&	-	\\						
					\midrule
					\multirow{4}{*}{100} & \multirow{4}{*}{200} & \multirow{2}{*}{1} & $\varepsilon_1$ &
					36.09	&	0.39	&	0.39	&	1.80	&	0.00	\\							
					& & & $\varepsilon_2$ &
					0.00	&	0.00	&	0.00	&	0.00	&	0.00	\\							
					\cmidrule{3-9}
					& & \multirow{2}{*}{0.25} & $\varepsilon_1$ &
					0.91	&	0.00	&	0.19	&	-	&	-	\\							
					& & & $\varepsilon_2$ &
					0.00	&	0.00	&	0.00	&	-	&	-	\\							
					\midrule
					\multirow{4}{*}{100} & \multirow{4}{*}{300} & \multirow{2}{*}{1} & $\varepsilon_1$ &
					35.83	&	0.31	&	0.31	&	2.18	&	0.00	\\							
					& & & $\varepsilon_2$ &	
					0.00	&	0.00	&	0.00	&	0.00	&	0.00	\\								
					\cmidrule{3-9}
					& & \multirow{2}{*}{0.25} & $\varepsilon_1$ &
					0.93	&	0.00	&	0.36	&	-	&	-	\\								
					& & & $\varepsilon_2$ &
					0.00	&	0.00	&	0.00	&	-	&	-	\\						
					\midrule
					\multirow{4}{*}{200} & \multirow{4}{*}{100} & \multirow{2}{*}{1} & $\varepsilon_1$ &
					78.46	&	0.00	&	0.00	&	3.22	&	0.08	\\							
					& & & $\varepsilon_2$ &
					0.00	&	0.00	&	0.00	&	0.06	&	0.08	\\							
					\cmidrule{3-9}
					& & \multirow{2}{*}{0.25} & $\varepsilon_1$ &
					11.51	&	1.35	&	3.60	&	38.00	&	-	\\						
					& & & $\varepsilon_2$ &
					1.34	&	0.34	&	1.31	&	27.33	&	-	\\						
					\midrule
					\multirow{4}{*}{200} & \multirow{4}{*}{200} & \multirow{2}{*}{1} & $\varepsilon_1$ &
					80.86	&	1.60	&	1.60	&	3.88	&	0.00	\\							
					& & & $\varepsilon_2$ &
					0.00	&	0.00	&	0.00	&	0.00	&	0.00	\\						
					\cmidrule{3-9}
					& & \multirow{2}{*}{0.25} & $\varepsilon_1$ &
					23.31	&	0.01	&	0.41	&	-	&	-	\\								
					& & & $\varepsilon_2$ &
					0.01	&	0.01	&	0.01	&	-	&	-	\\						
					\midrule
					\multirow{4}{*}{200} & \multirow{4}{*}{300} & \multirow{2}{*}{1} & $\varepsilon_1$ &
					78.86	&	0.00	&	0.45	&	2.78	&	0.00	\\								
					& & & $\varepsilon_2$ &
					0.00	&	0.00	&	0.00	&	0.00	&	0.00	\\							
					\cmidrule{3-9}
					& & \multirow{2}{*}{0.25} & $\varepsilon_1$ &
					48.20	&	0.00	&	0.00	&	-	&	-	\\							
					& & & $\varepsilon_2$ &
					0.00	&	0.00	&	0.00	&	-	&	-	\\							
					\bottomrule
					\multicolumn{9}{c}{\footnotesize Note: the dash ``-'' means there is no change-points detected.}
		\end{tabular}}} \label{tab 1-3}
	\end{table}
	
	Results provided in Tables \ref{tab 1-1}--\ref{tab 1-3} show that CP-TWSBM and CP-TWMNBS detections yield reliable estimates of the number of change-points and their locations. When $\rho = 1$, CP-TWAVG over-estimates the number of change-points, but it's interesting that it performs well for sparser case of $\rho = 0.25$. A possible explanation is that the sparser structure overcomes its inflated behavior to some extent. As for CP-GRA and CP-DMNBS detections, the performances of both methods are  reasonable in dense scenarios, especially CP-DMNBS. However, they are unable to detect any change-point for the sparser setting $\rho = 0.25$ in this example.
	
	\subsubsection{Three change-points detection}
	To assess the robustness of our method for change-point detection, we further construct a model with three change-points in the networks.
	We first design three types of link probability matrix changes, which we use to build dynamic networks later.  Given a link probability matrix $P$, define a changed link probability matrix initialized as $P' = P$. For two given sets $\cM_1, \cM_2 \subset \{1,2,\ldots,n\}$, for any $ i \in \cM_1$ and $j \in \cM_2$, the different types of link probability matrix changes are defined as follows:
	
	\begin{enumerate}[(1)]
		\item Coummunity switching: $P'_{i,\cdot} = P_{j,\cdot}, P'_{\cdot, i} = P_{\cdot, j}, P'_{j,\cdot} = P_{i,\cdot}, P'_{\cdot, j} = P_{\cdot, i}$.
		\item Community merging: $P'_{i,\cdot} = P_{j,\cdot}, P'_{\cdot, i} = P_{\cdot, j}$.
		\item Community changing: Regenerate $P'_{i, j}$ from Graphon 4.		
	\end{enumerate}
	
	Then the dynamic networks $\{A_t\}_{t = 1}^m$ for multiple change-points are designed as follows. $\cM_1$ and $\cM_2$ are two sets with $\lfloor n / 3\rfloor$ nodes randomly chosen from $\{1,2,\ldots,n\}$. For $t = 1, 2, \ldots, m / 4$, $A_t$ is generated from $P_1$ by Graphon 2. For $t = m / 4 + 1, \ldots, m / 2$, $A_t$ is generated from $P_2$ changed from $P_1$ by community switching. For $t = m / 2 + 1, \ldots, 3 m / 4$, $A_t$ is generated from $P_3$ changed from $P_2$ by community merging. For $t = 3 m / 4 + 1, \ldots, m$,  $A_t$ is generated from $P_4$ changed from $P_3$ by community changing. 
	The results are illustrated in Tables \ref{tab 3-1}--\ref{tab 3-3}.
	
	The reports suggest that CP-TWMNBS performs the best in terms of the number, efficiency and accuracy of change-point estimation. CP-TWSBM enjoys reasonably good behavior when $m = 100$ while encounters some false positives when $m $ increases to $200$. As for CP-TWAVG, although the estimated change-points number $\hat{\cJ}$ in Table \ref{tab 3-1} are not far away from real value $3$ and the efficient detect rates in Table \ref{tab 3-2} are all equal to $1$, the Boysen distances in Table \ref{tab 3-3} are sometimes too large to be accepted, i.e., the location error can not be controlled stably.

	\begin{table}[p]
		\centering
		\caption{Average estimated change-points number $\hat{J}$ under three change-points scenarios.}
		{\scalebox{0.75}{
				\begin{tabular}{*{8}{c}}
					\toprule		
					$m$ & $n$ & $\rho$ & CP-TWAVG & CP-TWSBM & CP-TWMNBS & CP-GRA & CP-DMNBS \\
					\midrule
					\multirow{2}{*}{100} & \multirow{2}{*}{100} & 1 &				
					3.00	&	3.00	&	3.00	&	0.31	&	3.00 \\												
					& & $0.25$ &
					1.92	&	3.40	&	2.99	&	0.00	&	0.02 \\								
					\midrule
					\multirow{2}{*}{100} & \multirow{2}{*}{200} & 1 &
					3.00	&	3.00	&	3.00	&	2.10	&	3.00 \\												
					& & $0.25$ &
					3.00	&	3.03	&	3.00	&	0.00	&	0.15 \\								
					\midrule
					\multirow{2}{*}{100} & \multirow{2}{*}{300} & 1 &
					3.00	&	3.00	&	3.00	&	0.00	&	3.00 \\												
					& & $0.25$ &
					3.00	&	3.00	&	3.00	&	0.00	&	1.52 \\							
					\midrule
					\multirow{2}{*}{200} & \multirow{2}{*}{100} & 1 &
					3.16	&	3.00	&	3.00	&	2.29	&	3.02 \\								
					& & $0.25$ &
					2.20	&	5.35	&	2.96	&	0.00	&	0.07 \\							
					\midrule
					\multirow{2}{*}{200} & \multirow{2}{*}{200} & 1 &
					3.37	&	3.00	&	3.00	&	1.06	&	3.01 \\							
					& & $0.25$ &
					3.00	&	4.57	&	3.01	&	0.00	&	1.62 \\									
					\midrule
					\multirow{2}{*}{200} & \multirow{2}{*}{300} & 1 &
					3.57	&	3.01	&	3.00	&	0.11	&	3.00 \\								
					& & $0.25$ &
					3.00	&	4.55	&	3.00	&	0.00	&	1.95 \\								
					\bottomrule
		\end{tabular}} }\label{tab 3-1}
	\end{table}

	\begin{table}[htbp]
		\centering
		\caption{Average efficient detect rate under three change-points scenarios.}
		{\scalebox{0.75}{
				\begin{tabular}{*{8}{c}}
					\toprule			
					$m$ & $n$ & $\rho$ & CP-TWAVG & CP-TWSBM & CP-TWMNBS & CP-GRA & CP-DMNBS \\
					\midrule	
					\multirow{2}{*}{100} & \multirow{2}{*}{100} & 1 &
					1.00	&	1.00	&	1.00	&	0.13	&	1.00	\\								
					& & $0.25$ &
					1.00	&	1.00	&	1.00	&	0.00	&	0.02	\\								
					\midrule
					\multirow{2}{*}{100} & \multirow{2}{*}{200} & 1 &
					1.00	&	1.00	&	1.00	&	1.00	&	1.00	\\								
					& & $0.25$ &
					1.00	&	1.00	&	1.00	&	0.00	&	0.15	\\								
					\midrule
					\multirow{2}{*}{100} & \multirow{2}{*}{300} & 1 &
					1.00	&	1.00	&	1.00	&	0.00	&	1.00	\\								
					& & $0.25$ &
					1.00	&	1.00	&	1.00	&	0.00	&	0.96	\\								
					\midrule
					\multirow{2}{*}{200} & \multirow{2}{*}{100} & 1 &
					1.00	&	1.00	&	1.00	&	0.81	&	1.00	\\								
					& & $0.25$ &
					1.00	&	1.00	&	1.00	&	0.00	&	0.07	\\								
					\midrule
					\multirow{2}{*}{200} & \multirow{2}{*}{200} & 1 &
					1.00	&	1.00	&	1.00	&	0.39	&	1.00	\\								
					& & $0.25$ &
					1.00	&	1.00	&	1.00	&	0.00	&	0.98	\\								
					\midrule
					\multirow{2}{*}{200} & \multirow{2}{*}{300} & 1 &
					1.00	&	1.00	&	1.00	&	0.06	&	1.00	\\							
					& & $0.25$ &
					1.00	&	1.00	&	1.00	&	0.00	&	1.00	\\								
					\bottomrule
		\end{tabular}}} \label{tab 3-2}
	\end{table}
	
	\begin{table}[htbp]
		\centering
		\caption{Average Boysen distances $\varepsilon_1, \varepsilon_2$ under three change-points scenarios.}
		{\scalebox{0.75}{
				\begin{tabular}{*{9}{c}}
					\toprule		
					$m$ & $n$ & $\rho$ & & CP-TWAVG & CP-TWSBM & CP-TWMNBS & CP-GRA & CP-DMNBS \\
					\midrule
					\multirow{4}{*}{100} & \multirow{4}{*}{100} & \multirow{2}{*}{1} & $\varepsilon_1$ &
					0.00	&	0.01	&	0.00	&	8.69	&	0.02	\\								
					& & & $\varepsilon_2$ &
					0.00	&	0.01	&	0.00	&	34.31	&	0.02	\\											
					\cmidrule{3-9}
					& & \multirow{2}{*}{0.25} & $\varepsilon_1$ &
					0.06	&	5.50	&	0.15	&	-	&	0.00	\\												
					& & & $\varepsilon_2$ &
					26.99	&	0.39	&	0.40	&	-	&	50.00	\\										
					\midrule
					\multirow{4}{*}{100} & \multirow{4}{*}{200} & \multirow{2}{*}{1} & $\varepsilon_1$ &
					0.00	&	0.00	&	0.00	&	1.17	&	0.00	\\											
					& & & $\varepsilon_2$ &
					0.00	&	0.00	&	0.00	&	25.10	&	0.00	\\											
					\cmidrule{3-9}
					& & \multirow{2}{*}{0.25} & $\varepsilon_1$ &
					0.00	&	0.45	&	0.00	&	-	&	0.07	\\											
					& & & $\varepsilon_2$ &
					0.00	&	0.00	&	0.00	&	-	&	50.07	\\									
					\midrule
					\multirow{4}{*}{100} & \multirow{4}{*}{300} & \multirow{2}{*}{1} & $\varepsilon_1$ &
					0.00	&	0.00	&	0.00	&	-	&	0.00	\\										
					& & & $\varepsilon_2$ &	
					0.00	&	0.00	&	0.00	&	-	&	0.00	\\											
					\cmidrule{3-9}
					& & \multirow{2}{*}{0.25} & $\varepsilon_1$ &
					0.00	&	0.00	&	0.00	&	-	&	0.03	\\												
					& & & $\varepsilon_2$ &
					0.00	&	0.00	&	0.00	&	-	&	34.90	\\									
					\midrule
					\multirow{4}{*}{200} & \multirow{4}{*}{100} & \multirow{2}{*}{1} & $\varepsilon_1$ &
					4.70	&	0.00	&	0.00	&	14.72	&	0.69	\\											
					& & & $\varepsilon_2$ &
					0.00	&	0.00	&	0.00	&	53.96	&	0.16	\\										
					\cmidrule{3-9}
					& & \multirow{2}{*}{0.25} & $\varepsilon_1$ &
					0.08	&	29.51	&	0.67	&	-	&	0.43	\\										
					& & & $\varepsilon_2$ &
					40.01	&	0.51	&	3.30	&	-	&	79.00	\\									
					\midrule
					\multirow{4}{*}{200} & \multirow{4}{*}{200} & \multirow{2}{*}{1} & $\varepsilon_1$ &
					11.27	&	0.00	&	0.00	&	15.72	&	0.35	\\											
					& & & $\varepsilon_2$ &
					0.00	&	0.00	&	0.00	&	61.03	&	0.04	\\										
					\cmidrule{3-9}
					& & \multirow{2}{*}{0.25} & $\varepsilon_1$ &
					0.00	&	27.64	&	0.30	&	-	&	0.04	\\												
					& & & $\varepsilon_2$ &
					0.00	&	0.03	&	0.00	&	-	&	65.83	\\									
					\midrule
					\multirow{4}{*}{200} & \multirow{4}{*}{300} & \multirow{2}{*}{1} & $\varepsilon_1$ &
					17.05	&	0.33	&	0.00	&	11.17	&	0.00	\\											
					& & & $\varepsilon_2$ &
					0.00	&	0.00	&	0.00	&	89.33	&	0.00	\\										
					\cmidrule{3-9}
					& & \multirow{2}{*}{0.25} & $\varepsilon_1$ &
					0.00	&	27.85	&	0.00	&	-	&	0.07	\\										
					& & & $\varepsilon_2$ &
					0.00	&	0.02	&	0.00	&	-	&	52.47	\\									
					\bottomrule
					\multicolumn{9}{c}{\footnotesize Note: the dash ``-'' means there is no change-points detected.}
		\end{tabular}}} \label{tab 3-3}
	\end{table}
	
	On the other hand, CP-GRA detection suffers greatly under-estimating the change-points, especially when $\rho =0.25$, there is no change-point detected in all cases. It happens similarly to CP-DMNBS detection when $\rho = 0.25$, so CP-DMNBS is also not the ideal for this scenario even though it is powerful when the networks are dense.
	
	Overall, the numerical experiments clearly demonstrate the superior performance of CP-TWMNBS detection over other detect methods for all simulation scenarios with CP-TWSBM method coming in second. CP-TWMNBS detection provides robust and stable performance across all experiments with more accurate $\hat{\cJ}$, higher efficient detection and smaller Boysen distances.

	\section{Data analysis}
	\label{sec-data}
	In this section, we analyze the performance of the proposed $TW_1$-type method for two-sample test and $TW_1$-type change-point detection using two real datasets. The first dataset used for the two-sample test comes from the Centers of Biomedical Research Excellence (COBRE) and the second dataset used for change-point detection is from MIT Reality Mining (RM) \citep{eagle2009}.

	\subsection{Two-sample test with real data example}
	Raw anatomical and functional scans from 146 subjects of 72 patients with schizophrenia (SCZ) and 74 healthy controls (HCs) can be downloaded from a public database (\url{http://fcon_1000.projects.nitrc.org/indi/retro/cobre.html}). In this paper, we use the processed connectomics dataset in \cite{arroyo2017}. After a series of pre-processing steps, \cite{arroyo2017} keeps 54 SCZ and 70 HC subjects for analysis and chooses 264 brain regions of interest as the nodes. For each of the 263 nodes with every other node, they applies Fisher's R-to-Z transformation to the cross-correlation matrix of Pearson $r$-values. 
	
	In our study, we perform the Z-to-R inverse transformation to their dataset to get the original cross-correlation matrix of Pearson $r$-values, which is denoted as $R$. To analyze graphical properties of these brain functional networks, we need to create an adjacency matrix $A$ from $R$. We set $A_{i j}$ to be $1$ if $R_{i j}$ exceeds a threshold $T$ and $A_{i j}$ to be $0$ otherwise. There is no generally accepted way to identify an optimal threshold for this graph construction procedure, we decide to set $T$ varied between $0.3$ and $0.7$ with step of $0.05$.
	
	For each threshold $T$, two situations are considered for the two-sample test. In the first situation, we randomly divide HC into $2$ groups with sample sizes $m_1 = m_2 = 35$ and calculate the average null hypothesis reject rates of $TW_1$-MNBS test, $TW_1$-AVG test, $TW_1$-SBM test, $\chi^2$-type test, and $N$-type test through 100 repeated simulations. In the second situation, we apply the same test methods above to two groups of SCZ and HC directly and compare their average null hypothesis reject rates. In both cases, the significance level is set to be 0.05. The results are shown in Tables \ref{tab p-value} and \ref{tab power} respectively.
	
	\begin{table}[htbp]
		\centering
		\caption{Average $H_0$ reject rate of test over HC group over 100 simulations.}
		{\scalebox{0.75}{\begin{tabular}{*{10}{c}}
					\toprule			    
					$T$ & $0.30$ & $0.35$ & $0.40$ & $0.45$ & $0.50$ & $0.55$ & $0.60$ & $0.65$ & $0.70$ \\
					\midrule
					$TW_1$-AVG & 
					$1$ & $1$ & $1$ & $1$ & $0.80$ & $0.67$ & $0.57$ & $0.50$ & $0.44$ \\
					$TW_1$-SBM & 
					$1$ & $1$ & $1$ & $1$ & $1$ & $1$ & $1$ & $1$ & $1$ \\
					$TW_1$-MNBS & 
					$1$ & $1$ & $1$ & $1$ & $1$ & $0$ & $0$ & $1$ & $1$ \\
					$\chi^2$-type & 
					$1$ & $1$ & $0$ & $0$ & $0$ & $0$ & $0$ & $0$ & $0$ \\
					$N$-type & 
					$1$ & $1$ & $1$ & $0$ & $0$ & $0$ & $0$ & $0$ & $0$ \\	
					\bottomrule
		\end{tabular}} }\label{tab p-value}
	\end{table}
	
	\begin{table}[htbp]
		\centering
		\caption{Average $H_0$ reject rate of test over SCZ and HC groups.}
		{\scalebox{0.75}{\begin{tabular}{*{10}{c}}
					\toprule			    
					$T$ & $0.30$ & $0.35$ & $0.40$ & $0.45$ & $0.50$ & $0.55$ & $0.60$ & $0.65$ & $0.70$ \\
					\midrule					
					$TW_1$-AVG & 
					$1$ & $1$ & $1$ & $1$ & $1$ & $0$ & $0$ & $0$ & $0$ \\
					$TW_1$-SBM & 
					$1$ & $1$ & $1$ & $1$ & $1$ & $1$ & $1$ & $1$ & $1$ \\
					$TW_1$-MNBS & 
					$1$ & $1$ & $1$ & $1$ & $1$ & $1$ & $1$ & $1$ & $1$ \\
					$\chi^2$-type & 
					$0$ & $0$ & $0$ & $0$ & $0$ & $0$ & $0$ & $0$ & $0$ \\	
					$N$-type & 
					$1$ & $1$ & $1$ & $1$ & $1$ & $1$ & $1$ & $1$ & $1$ \\					
					\bottomrule
		\end{tabular}}} \label{tab power}
	\end{table}
	
	To investigate the performance of the tests, we need to compare the type \rom{1} error in Table \ref{tab p-value} and the power result in Table \ref{tab power} together. Table \ref{tab p-value} shows that $TW_1$-type tests based on SBM and AVG have poor performance for the test over HC group because the reject rates all exceed $0.05$ and even equal to $1$. From Table \ref{tab power}, it is found that $\chi^2$-type test loses power for the test over SCZ and HC groups, where the reject rates are all $0$. Only $TW_1$-type test based on MNBS when $T = 0.55, 0.60$ and $N$-type test when $T \geq 0.45$ can perform well in both situations. In addition, applying MNBS, we illustrate the adjacency matrices of subject-specific networks of HC and SCZ groups when $T = 0.60$ in Figure \ref{adjacency}. One can find that the two groups do have differences in the network structure.
	
	\begin{figure}[htbp]
		\centering
		\begin{subfigure}{0.3\columnwidth}
			\centering
			\includegraphics[width = \linewidth]{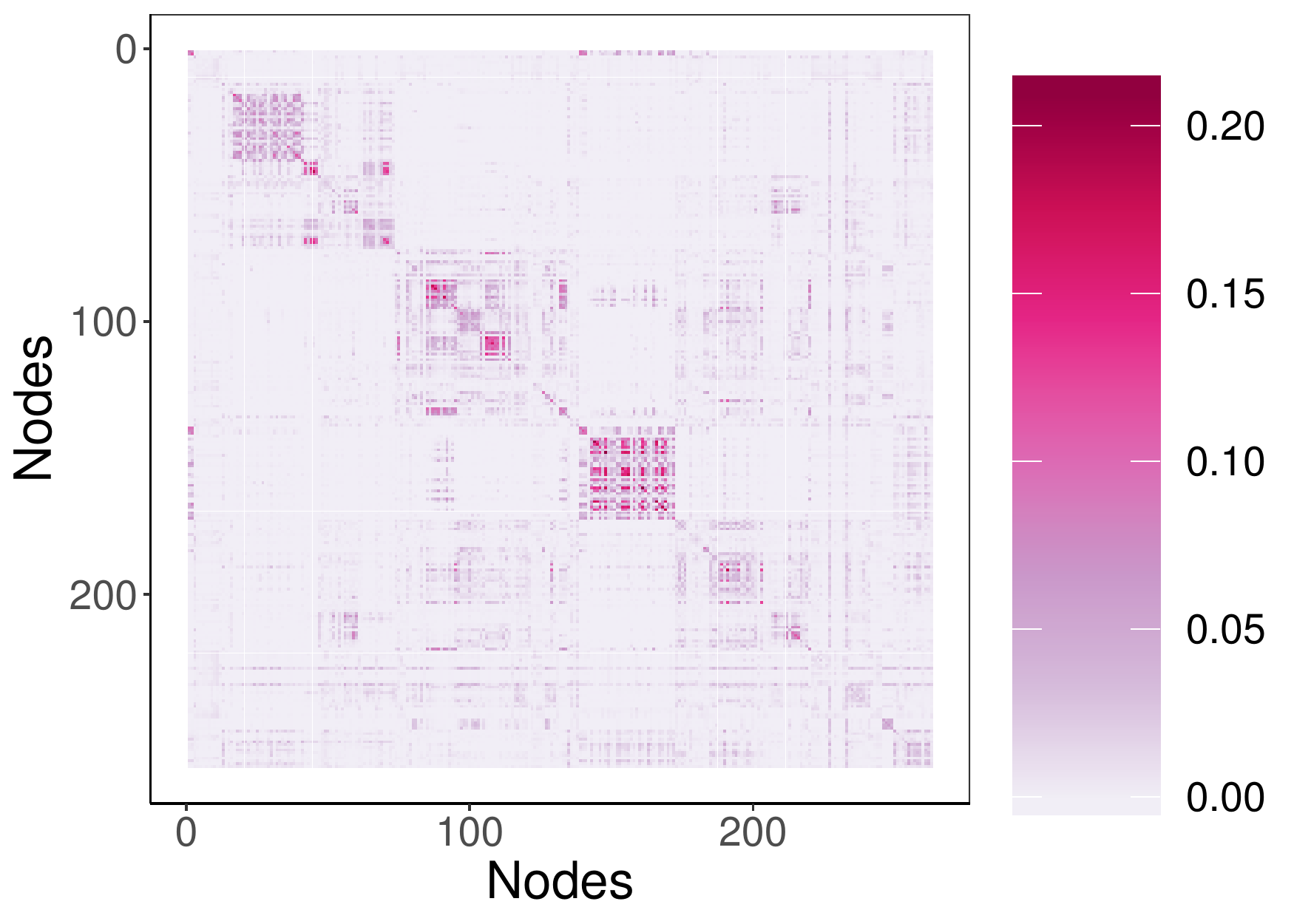} 
			\caption{HC} \label{healthy}
		\end{subfigure}
		\begin{subfigure}{0.3\columnwidth}
			\centering
			\includegraphics[width = \linewidth]{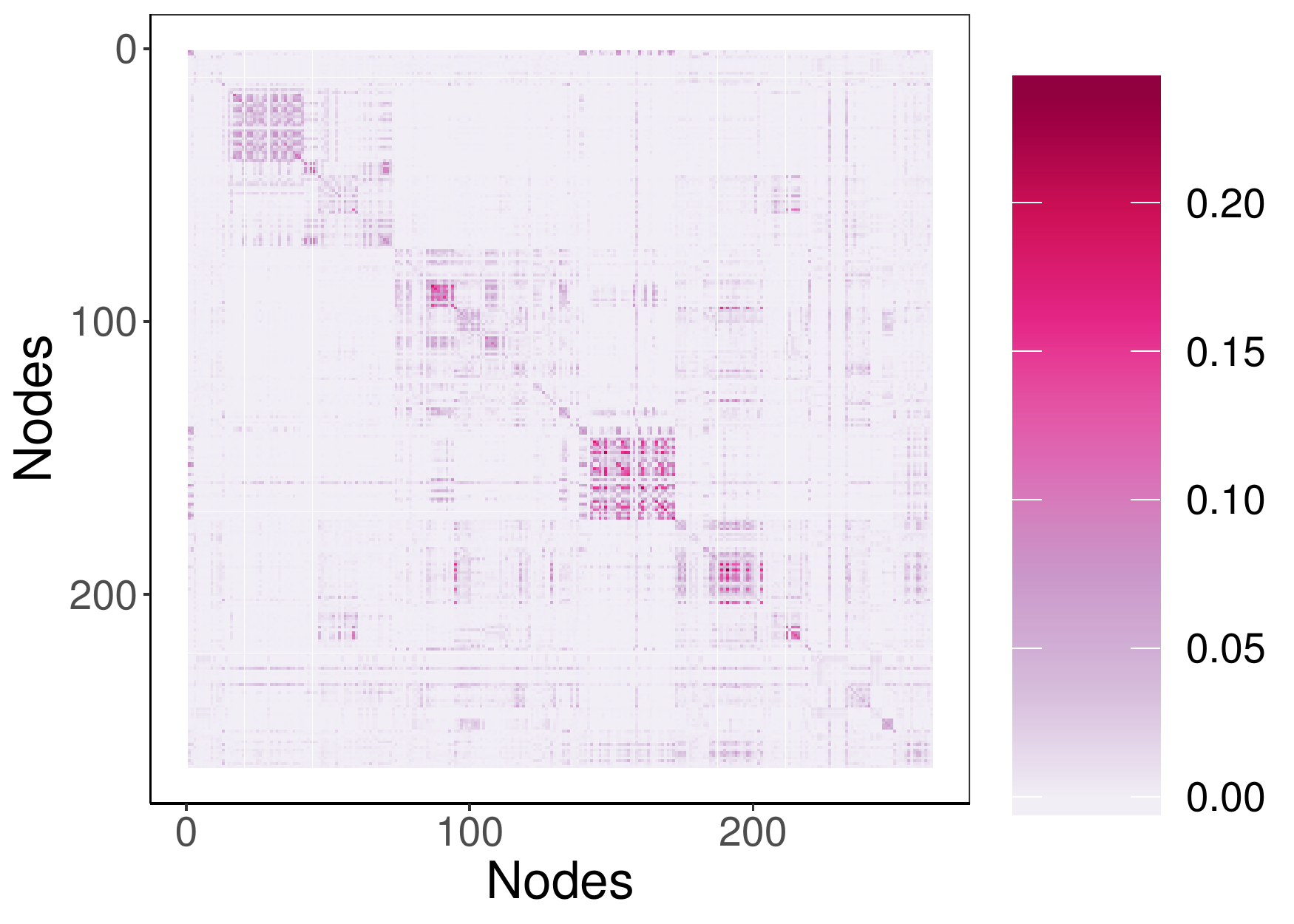} 
			\caption{SCZ} \label{schizophrenia}
		\end{subfigure}
		\caption{Adjacency matrices estimated by MNBS for HC and SCZ groups.}
		\label{adjacency}
	\end{figure}
	
	\subsection{Change-point detection in dynamic networks}
	In this section, we apply CP-TWMNBS, CP-TWAVG, CP-TWSBM, CP-GRA, and CP-DMNBS detections to perform change-point detection for a phone-call network data extracted from RM dataset. The data is collected through an experiment conducted by the MIT Media Laboratory following 106 MIT students and staff using mobile phones with preinstalled software that can record and send call logs from 2004 to 2005 academic year. Note that this is different from the MIT proximity network data considered in \cite{zhao2019} which is based on the bluetooth scans instead of phone calls.  In this analysis, we are interested in whether phone call patterns changed during this time, which may reflect a change in relationship among these subjects. 94 of the 106 RM subjects completed the survey, we remain records only within these participants and filter records before $07/20/2004$ due to the extreme scarcity of sample before that time. Then there remains  81 subjects left and we construct dynamic networks among these subjects by day. For each day, construct a network with the subjects as nodes and a link between two subjects if they had at least one call on that day. We encode the network of each day by an adjacency matrix, with $1$ for element $(i, j)$ if there is an edge between subject $i$ and subject $j$, and $0$ otherwise. Thus, there are in total 310 days from $07/20/2004$ to $06/14/2005$. The calendar of events is included in the appendix. 
	We claim that an estimated change-point is reasonable if it is at most three days away from the real dates the event lasts.

	We first choose $h = 7$ and Figure \ref{identical-result-1} plots the results of different methods on the dynamic networks. The purple shadow areas mark time intervals from the beginning to the end of events continue on MIT academic calendar 2004--2005, which  can be used as references for the estimated change-points' occurrences. The red lines in Figure \ref{identical-result-1} are the estimated change-points applying different detect methods. 
	
	\begin{figure}[htbp]
		\centering
		\includegraphics[width=11cm]{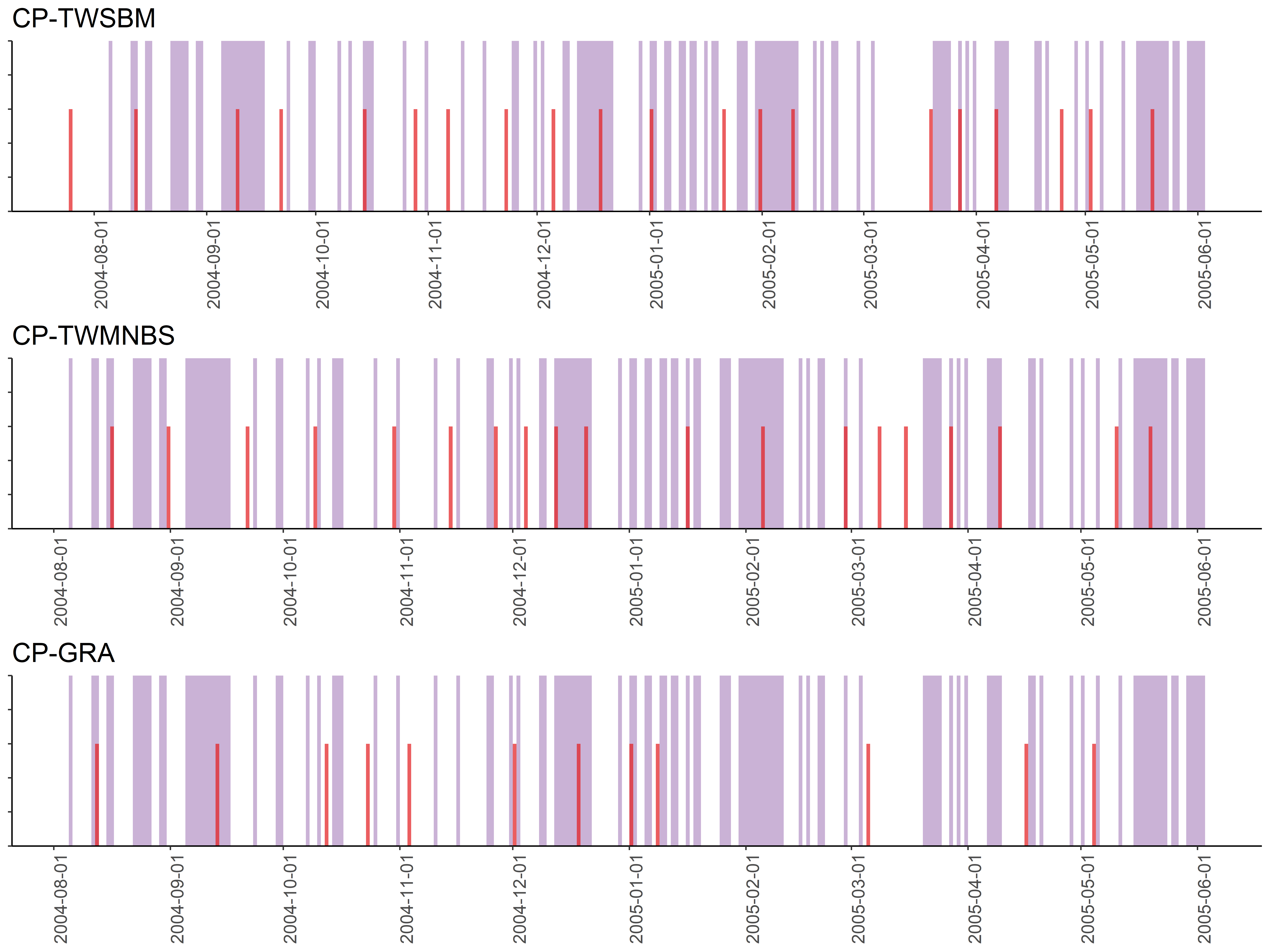} 
		\caption{Calendar time intervals of events and estimated change-points.} \label{identical-result-1}
	\end{figure}
	
	It turns out that CP-TWAVG and CP-DMNBS detections either do not work well or detect no change-point. CP-TWSBM method detects 20 change-points, CP-TWMNBS method detects 19 change-points,  while CP-GRA detection detects 12 change-points. When comparing the estimated change-points to intervals of calendar events, we see that they align each other the best by using CP-TWMNBS detection and then CP-TWSBM detection, whereas there are more estimated change-points by CP-GRA detection  that can not be explained.

	However, it's observed that some of the change-points detected by CP-TWSBM and CP-TWMNBS methods can be a  little trivial. For example, CP-TWMNBS detected a change-point occurred at around 01/09/2004, which is near event ``English Evaluate Test for International Students'' in the calendar. To ignore the less significant events, we only consider the seemingly major events displayed in bold in the calendar as possible reasons for estimated change-points and set $h = 14$, which corresponds to 2 weeks. The details are reported in Table \ref{tab cp}. The CP-TWMNBS and CP-TWSBM methods detect 9 change-points, CP-GRA method detects 13 change-points. Notably CP-GRA method still labels more trivial change-points away from the important events. Based on the results, it is most likely valid in saying that CP-TWSBM and CP-TWMNBS detections are more reliable.

	\begin{table}[htbp]
		\centering
		\caption{Estimated change-points by different methods for MIT phone data.}
		{\scalebox{0.75}{\begin{tabular}{*{8}{c}}
					\toprule			    
					CP-TWSBM & 02/08/2004 & 08/09/2004  & 12/10/2004 &  15/11/2004 &  27/12/2004  & 03/02/2005 & 19/02/2005\\
					&  23/03/2005 &  03/05/2005 &&&&&\\
					\midrule
					CP-TWMNBS & 10/08/2004 & 27/09/2004 & 17/10/2004 &  20/11/2004 & 05/12/2004 & 24/12/2004 & 13/02/2005\\
					& 10/04/2005 & 04/05/2005 &&&&& \\
					\midrule
					CP-GRA & 13/08/2004 & 01/09/2004 & 14/09/2004 & 28/09/2004 & 13/10/2004 & 04/11/2004 & 18/11/2004\\
					& 02/12/2004 & 19/12/2004 & 09/01/2005 & 06/03/2005 & 17/04/2005 & 05/05/2005 &\\
					\bottomrule
		\end{tabular}}} \label{tab cp}
	\end{table}

	\section{Conclusion}
	We consider the problem of hypothesis testing on whether two populations of networks defined on a common vertex set are from the same distribution. Two-sample testing on populations of networks is a challenging task especially when the the number of nodes  is large. We propose a general $TW_1$-type test (which is later adapted to a change-point detection procedure in dynamic networks), derive its asymptotic distribution and asymptotic power. The test statistics utilizes some plugin estimates for the link probability matrices and properties of the resulting tests  with various estimates are discussed by  evaluating and comparing $TW_1$-type tests based on MNBS, AVG, SBM  theoretically, and numerically with both simulated and real data. From the simulation study, we see that the proposed $TW_1$-type test based on MNBS performs the best and yields robust results even when the structure is sparse. In addition, we provide a significant modification of the two-sample network test for change-point detection in dynamic networks. Simulation and real data analyses show that the procedure is consistent, principled and practically viable.

	\section*{Acknowledgements}
	\vspace{-.5em}
	The work of Li Chen was supported by the China Scholarship Council under Grant 201806240032 and the Fundamental Research Funds for the Central Universities, Southwest Minzu University under grant 2021NQNCZ02. 
	The work of Jie Zhou was supported in part by the National Natural Science Foundation of China under grants 61374027 and 11871357, and in part by the Sichuan Science and Technology Program under grant 2019YJ0122. 
	Lizhen Lin acknowledges the generous support from NSF grants IIS 1663870, DMS Career 1654579, DMS 2113642 and a DARPA grant N66001-17-1-4041.

	\appendix
	
	The appendix mainly includes theorem proofs  omitted and the academic calendar of MIT we use in the paper.
	
	\section{Preliminaries}
	\begin{proposition}[Hoeffding's inequality \citep{hoeffding1963}]
		\label{Hoeffding}
		If $X_1, X_2,\ldots, X_m$ are independent random variables and $a_i \leq X_i \leq b_i (i = 1,2, \ldots,m)$, then for $t > 0$, 
		$$P\left(\bar{X} - \mu \geq t\right)\leq \exp\left\{-\frac{2 m^2 t^2}{\sum_{i = 1}^{m} (b_i - a_i)^2}\right\},$$
		where $\bar{X} = \frac{1}{m}\sum_{i=1}^{m}X_i,  \ \mu = E(\bar{X}).$
	\end{proposition}

	\begin{proposition}[Bernstein's inequality \citep{Bernstein}]
		\label{Bernstein}
		Let $X_1, X_2,\ldots, X_m$ be independent zero-mean random variables. Suppose that $|X_i|\leq M$ with probability $1$ for all $i$. Then for all positive $t$, we have
		$$P\left(\sum_{i=1}^{m}X_i>t\right)\leq \exp\left\{-\frac{\frac{1}{2}t^2}{\sum_{i=1}^{m}E(X_i^2)+\frac{1}{3}Mt}\right\}.$$
	\end{proposition}
	
	For a sequence of independent Bernoulli random variables where $X_i \sim \text{Bernoulli}(p)$, by Proposition \ref{Hoeffding} we have
	$$P\big(\left|\bar{X} - p \right| \geq t\big)\leq 2 \exp\big\{-2 m t^2\big\}.$$
	Similarly, by Proposition \ref{Bernstein}, we have
	$$P\big(\left|\bar{X} - p \right|>t\big)\leq 2\exp\left\{ -\frac{\frac{1}{2} m t^2}{p (1 - p) + \frac{1}{3} t}\right\}.$$
	
	\begin{lemma}[Asymptotic distributions of $\lambda_1(Z)$ and $\lambda_n(Z)$]
		\label{Asymptotic distributions of Z_star for two-sample}
		For $Z$ defined in \eqref{Z for two-smaple} in subsection \ref{subsection General model}, we have
		\begin{equation*} 
		n^{2/3} [\lambda_1 (Z) - 2] \rightsquigarrow TW_1, ~
		n^{2/3} [- \lambda_n (Z) - 2] \rightsquigarrow TW_1.
		\end{equation*}
	\end{lemma}
	
	\begin{proof}
		Let $G$ be an $n \times n$ symetric matrix whose upper diagonal entries are independent normal with mean zero and variance $1 / (n - 1)$, and zero diagonal entries. Let $H_G = \sqrt{(n - 1) / n} G$, according to Theorem 1.2 in \cite{lee2014}, $n^{2/3}[\lambda_1(H_G) - 2]$ converges to $TW_1$ in distribution. For convenience and without ambiguity, we also use $TW_1$ to denote a random variable following the Tracy--Widom law with index 1. Then we have
		\begin{equation*}
		\lambda_1(H_G) = 2 + n^{-2/3} TW_1 + o_p(n^{- 2/3}).
		\end{equation*}
		Further,
		\begin{equation*}
		\lambda_1(G) = \sqrt{\frac{n}{n - 1}} \lambda_1(H_G)  = \big[ 1 + O_n\big( n^{- 1}\big) \big]  \lambda_1(H_G) = 2 +  n^{-2/3} TW_1 + o_p(n^{- 2/3}),
		\end{equation*}		
		which is equivalent to 
		\begin{equation*}
		n^{2 / 3} [\lambda_1(G) - 2] \rightsquigarrow TW_1.
		\end{equation*}
		
		Since the first and second moments of entries of $Z$ and $G$ are the same, it follows from Theorem 2.4 in \cite{erdHos2012} that $n^{2/3}[\lambda_1(Z) - 2]$ and $n^{2/3}[\lambda_1(G) - 2]$ have the same limiting distribution. Therefore, $$n^{2/3} [\lambda_1 (Z) - 2] \rightsquigarrow TW_1.$$
		The same argument applies to $\lambda_n (Z)$.
		
	\end{proof}
	
	\section{Proof of Theorem \ref{asymptotic null distribution result for GENERAL}}
	Under the null hypothesis $H_0$, we have $P_1 = P_2 \equiv P$, and it's not difficult to observe that
	\begin{equation} \label{Z_bar and Z^star GENERAL}
	\hat{Z}_{i j} = \frac{\sqrt{\frac{1}{m_1} P_{i j} (1 - P_{i j}) + \frac{1}{m_2} P_{i j} (1 - P_{i j})}}{\sqrt{\frac{1}{m_1} \hat{P}_{1, i j} \big(1 - \hat{P}_{1, i j}\big) + \frac{1}{m_2} \hat{P}_{2, i j} \big(1 - \hat{P}_{2, i j}\big)}} Z_{i j}.
	\end{equation}
	
	Since
	\begin{equation} \label{order of P_hat GENERAL}
	\sup_{i, j} |\hat{P}_{u, i j} - P_{i j}| = o_p(n^{- 2 / 3}),
	\end{equation}
	for the numerator in \eqref{Z_bar and Z^star GENERAL}, utilizing the Taylor Expansion, we have	
	\begin{align*} 
	& \quad \sqrt{\frac{1}{m_1} P_{i j} (1 - P_{i j}) + \frac{1}{m_2} P_{i j} (1 - P_{i j})} \\
	& = \sqrt{\frac{m_1 + m_2}{m_1 m_2}  P_{i j} (1 - P_{i j})} \\
	& = \sqrt{\frac{m_1 + m_2}{m_1 m_2}} \Big[ \sqrt{\hat{P}_{1, i j} \big(1 - \hat{P}_{1, i j}\big)} + O_n\big(P_{i j} - \hat{P}_{1, i j}\big)\Big] \\ 
	& =  \sqrt{\frac{m_1 + m_2}{m_1 m_2}} \Big[ \sqrt{\hat{P}_{1, i j} \big(1 - \hat{P}_{1, i j}\big)} + o_p( n^{- 2 / 3}) \Big] \\
	& = \sqrt{\frac{1}{m_1} \hat{P}_{1, i j} \big(1 - \hat{P}_{1, i j}\big) + \frac{1}{m_2} \hat{P}_{1, i j} \big(1 - \hat{P}_{1, i j}\big)} + \sqrt{\frac{m_1 + m_2}{m_1 m_2}}  o_p(n^{- 2 / 3}),
	\end{align*}
	where the third equality is obtained by condition \eqref{order of P_hat GENERAL}.
	
	Without loss of gernerality, assume $\hat{P}_{1, i j} \big(1 - \hat{P}_{1, i j}\big) \leq \hat{P}_{2, i j} \big(1 - \hat{P}_{2, i j}\big)$, then
	\begin{align} \label{approx1}
	& \quad \sqrt{\frac{1}{m_1} P_{i j} (1 - P_{i j}) + \frac{1}{m_2} P_{i j} (1 - P_{i j})}  \\
	& \leq \sqrt{\frac{1}{m_1} \hat{P}_{1, i j} \big(1 - \hat{P}_{1, i j}\big) + \frac{1}{m_2} \hat{P}_{2, i j} \big(1 - \hat{P}_{2, i j}\big)} + \sqrt{\frac{m_1 + m_2}{m_1 m_2}} o_p(n^{- 2 / 3}). \nonumber
	\end{align}	
	Similarily, we have
	\begin{equation} \label{approx2}
	\begin{split}		
	& \quad \sqrt{\frac{1}{m_1} \hat{P}_{1, i j} \big(1 - \hat{P}_{1, i j}\big) + \frac{1}{m_2} \hat{P}_{2, i j} \big(1 - \hat{P}_{2, i j}\big) }  \\
	& \leq \sqrt{\frac{1}{m_1} P_{i j} (1 - P_{i j}) + \frac{1}{m_2} P_{i j} (1 - P_{i j})} + \sqrt{\frac{m_1 + m_2}{m_1 m_2}} o_p( n^{- 2 / 3}). 
	\end{split}
	\end{equation}
	From \eqref{approx1} and \eqref{approx2}, 
	\begin{equation} \label{approx GENERAL}
	\begin{split}
	& \quad \sqrt{ \frac{1}{m_1} P_{i j} (1 - P_{i j}) + \frac{1}{m_2} P_{i j} (1 - P_{i j})}  \\
	& = \sqrt{\frac{1}{m_1} \hat{P}_{1, i j} \big(1 - \hat{P}_{1, i j}\big) + \frac{1}{m_2} \hat{P}_{2, i j} \big(1 - \hat{P}_{2, i j}\big)} + \sqrt{ \frac{m_1 + m_2}{m_1 m_2} } o_p(n^{- 2 / 3}). 
	\end{split}
	\end{equation}

	Combining \eqref{approx GENERAL} with \eqref{Z_bar and Z^star GENERAL}, we have
	\begin{equation} \label{Z hat Z}
	\hat{Z} - Z = M \circ Z,
	\end{equation}
	where $M$ is an $n \times n$ matrix whose elements $M_{i j} = o_p(n^{- 2 / 3})$ and the notation $\circ$ denotes the Hadamard (element-wise) product of two matrices.
	
	One has
	\begin{equation}
	\begin{split}
	\|\hat{Z} - Z\|_{op} & = \|M \circ Z\|_{op} \\
	& = \sup_{\substack{\|x\|_2 = 1, \\ x \in \R^n}} \|(M \circ Z) x\|_2 \\
	& = \sup_{\substack{\|x\|_2 = 1, \\ x \in \R^n}} \sqrt{\sum_{i = 1}^n \left(\sum_{j = 1}^n M_{i j} Z_{i j} x_j\right)^2},\\
	&  = \sqrt{\sum_{i = 1}^n \left(\sum_{j = 1}^n M_{i j} Z_{i j} x^*_j\right)^2}, \label{operator 2}
	\end{split}
	\end{equation}
	where $\norm \cdot _{op}$ denotes the operator norm of a matrix, $\|\cdot\|_2$ is the Euclidean norm of a vector, and $x^*$ is a unit eigenvector of the largest singular of $M \circ Z$. 
	
	Define an $n \times n$ symmetric matrix as $M^*$ and a unit vector as $x'$. Consider the last equality of \eqref{operator 2}, let 
		\begin{equation*}
		M^*_{i j} = 
		\begin{cases}
		M_{i j}, \ & Z_{i j} x^*_j \geq 0, \\
		- M_{i j}, \ & Z_{i j} x^*_j < 0,
		\end{cases}		
		\end{equation*} 
	
	and 
		\begin{equation*}
		x'_{j} = 
		\begin{cases}
		x^*_j, \ & Z_{i j} x^*_j \geq 0, \\
		- x^*_j, \ & Z_{i j} x^*_j < 0.
		\end{cases}	
		\end{equation*}
	Therefore, we have
		\begin{equation}
		\begin{split}
		\|\hat{Z} - Z\|_{op} & = \sqrt{\sum_{i = 1}^n \left(\sum_{j = 1}^n M^*_{i j} Z_{i j} x'_j\right)^2}  \\
		& \leq  (\sup_{\substack{i,j}} |M^*_{i j}|) \sqrt{\sum_{i = 1}^n \left(\sum_{j = 1}^n Z_{i j} x'_j\right)^2}  \\
		& \leq (\sup_{\substack{i,j}} |M^*_{i j}|) \|Z\|_{op}. \label{eq}
		\end{split}	
		\end{equation}
		The first inequality in \eqref{eq} holds true since $Z_{i j}^* x'_j$ are non-negative for all $i$ and $j$. In addition, $Z$ is a Wigner matrix and from Corollary 2.3.6 in \cite{tao2012}, the norm of $Z$ satisties
	\begin{equation*}
	\|Z\|_{op} = O_p(1).
	\end{equation*}
	The meaning of notation $O_p(\cdot)$ is as follows: For two sequences of real numbers $\{x_n\}$ and $\{y_n\}$, we write $y_n = O_p(x_n)$ if for any $\varepsilon > 0$, there exist finite $C > 0$ and $N > 0$ such that $P\big( \big| \frac{y_n}{x_n} \big| > C \big) < \varepsilon$ for any $ n > N$.
	
	It is noted that $M^*_{i j} = o_p(n^{- 2 / 3})$, so	
	\begin{equation*}
	\|\hat{Z} - Z\|_{op} \leq o_p(n^{- 2 / 3}).
	\end{equation*}
	Then
	\begin{equation} \label{lambda_0 GENERAL}
	| \lambda_1 (\hat{Z}) - \lambda_1 (Z) | \leq o_p(n^{- 2/3}).
	\end{equation}
	Combining \eqref{lambda_0 GENERAL} with Lemma \ref{Asymptotic distributions of Z_star for two-sample}, we have
	\begin{equation*} \label{lambda_1 GENERAL}
	n^{2/3} [\lambda_1 (\hat{Z}) - 2] \rightsquigarrow TW_1.
	\end{equation*}
	Similarly, we can prove
	\begin{equation*}  \label{lambda_2 GENERAL}
	n^{2/3} [- \lambda_n (\hat{Z}) - 2] \rightsquigarrow TW_1.
	\end{equation*}
	
	
	
	\section{Proof of Corollary \ref{asymptotic type \rom{1} error for GENERAL}}
	\begin{align*}
	P(T_{TW_1} \geq \tau_{\alpha / 2}) & \leq P \big(n^{2 / 3} [\lambda_1(\hat{Z}) - 2] \geq \tau_{\alpha / 2}\big) + P\big(n^{2 / 3} [- \lambda_n(\hat{Z}) - 2] \geq \tau_{\alpha / 2}\big) \\
	& = \alpha / 2 + o_n(1) + \alpha / 2 +o_n(1)\\
	& = \alpha + o_n(1).
	\end{align*}
	
	
	\section{Proof of of Corollary \ref{asymptotic power guarantee for GENERAL}}
	Define a matrix $W \in \R^{n \times n}$ with zero diagonal and for any $i \neq j$,
	$$W_{i j} = \frac{(P_{1, i j} - P_{2, i j}) - (\bar{A}_{1, i j} - \bar{A}_{2, i j})}{\sqrt{ (n - 1) \Big[ \frac{1}{m_1}P_{1, i j} ( 1 - P_{1, i j}) + \frac{1}{m_2}P_{2, i j} ( 1 - P_{2, i j})\Big]  }}.$$
	Recall the definitions of $Z, \hat{Z}$ and $\tilde{Z}$ given by \eqref{Z for two-smaple}, \eqref{Z_hat for two-sample} and \eqref{Z tuta} {in subsection \ref{subsection General model} respectively, from \eqref{Z hat Z}, it is easy to get
	$$\hat{Z}_{i j} = [ 1 + o_p(n^{-2/3})] (\tilde{Z}_{i j} - W_{i j}).$$
	Thus
	$$[ 1 + o_p(n^{-2/3})] \tilde{Z}_{i j} = \hat{Z}_{i j} + [ 1 + o_p(n^{-2/3})] W_{i j}.$$

	This implies
	$$\hat{Z} = \tilde{Z} \circ (J + D) - W \circ (J + D),$$
	where $J$ is an $n \times n$ matrix with every element equal to 1, and $D$ is an $n \times n$ matrix with elements $D_{i j} = o_p(n^{- 2 / 3})$. Similarly with the proof of Theorem \ref{asymptotic null distribution result for GENERAL}, we can get $\sigma_1(\tilde{Z} \circ (J + D)) = \sigma_1(\tilde{Z})$,  $\sigma_1(W \circ (J + D)) =  \sigma_1(W)$ with probability 1 as $n \to \infty$. 
	
	Applying the triangle inequality of spectral norm, we have
	$$\sigma_1(\hat{Z}) \geq\sigma_1{\tilde{(Z)}} - \sigma_1{(W)}$$ 		
	with probability $1$ as $n$ tends to infinity. Noting that $W$ is a mean zero matrix whose singular value can be bounded by using the $TW_1$ asymptotic distribution. Hence, for any $\beta \in (0, 1)$,
	\begin{equation}  \label{W}
	\begin{split}
	& \quad P\big(\sigma_1(W) \leq 2 + n^{- 2 / 3} \tau_{\beta}\big) \\
	& = 1 - P\big(\sigma_1(W) > 2 + n^{- 2 / 3} \tau_{\beta}\big) \\
	& \geq 1 - \big[P\big(\lambda_1(W) > 2 + n^{- 2 / 3} \tau_{\beta}\big) + P\big( - \lambda_n(W) > 2 + n^{- 2 / 3} \tau_{\beta}\big)\big] \\
	& = 1 - 2 \beta + o_n(1).
	\end{split}
	\end{equation}
	Set $\tau_{\beta}  =  n^{2 / 3}[ \sigma_1(\tilde{Z}) - 4] - \tau_{\alpha / 2}$, and plug this in \eqref{W}, then we have 
	\begin{align*}
	1 - 2 \beta + o_n(1) & \leq P\left( \sigma_1{(W)} \leq 2 + n^{- 2 / 3} \big\{n^{2 / 3}[ \sigma_1(\tilde{Z}) - 4] - \tau_{\alpha / 2}\big\}\right) \\
	& = P\left( 2 + n^{- 2 / 3} \tau_{\alpha / 2} \leq \sigma_1(\tilde{Z}) - \sigma_1{(W)}\right) \\
	& \leq P\left( 2 + n^{- 2 / 3} \tau_{\alpha / 2} \leq \sigma_1(\hat{Z})\right) \\
	& = P\big(n^{2 / 3} [\sigma_1(\hat{Z}) - 2] \geq \tau_{\alpha / 2}\big) \\
	& = P(T_{TW_1} \geq \tau_{\alpha / 2}).
	\end{align*}
	Observe that if $n^{- 2 / 3}[ \sigma_1(\tilde{Z}) - 4]^{- 1} \leq o_n(1)$, for a fixed $\alpha \in (0, 1)$, we have $\tau_{\beta}^{- 1} = o_n(1)$, that is $\beta = o_n(1)$. Therefore,
	$$P(T_{TW_1} \geq \tau_{\alpha / 2}) = 1 + o_n(1).$$

	\section{Proof of Theorem \ref{consistency of change-points detection}}
	For any $t$ that is not a change-point, since $t \in \cJ$, we have		
	\begin{align*}
	P(T_{TW_1}(t, h) > \vartriangle_{T_{TW_1}}) & = 1 -  P(T_{TW_1}(t, h) \leq \vartriangle_{T_{TW_1}}) \\
	& = 1 - \prod_{t' \in (t - h, t + h)} P(T_{TW_1}(t', h) \leq \vartriangle_{T_{TW_1}}) \\
	& = 1 - \prod_{t' \in (t - h, t + h)} \big[1 -  P(T_{TW_1}(t', h) > \vartriangle_{T_{TW_1}})\big] \\
	& \leq 1 - \prod_{t' \in (t - h, t + h)} \big[1 -  P(T_{TW_1}(t', h) > \tau_{\alpha} )\big]  \\
	& \leq 1 - (1 - 2 \alpha)^{2 h} + o_n(1) \\
	& = 1/n + o_n(1) \rightarrow 0.
	\end{align*}
	
	For any $t$ that is a true change-point, under the alternative hypothesis, $n^{2 / 3} [\delta(t, h) - 4] \geq 2 \tau_{\alpha}$. We have
	\begin{align} \label{a}
	\begin{split}		
	P(T_{TW_1}(t, h) > \vartriangle_{T_{TW_1}}) & = P\Big(n^{2 / 3} \big[\sigma_1 \big(\hat{Z}(t, h)\big) - 2\big] > n^{2 / 3} [\delta(t, h) - 4] - \tau_{\alpha}\Big) \\
	& = P\big(\sigma_1 \big(\hat{Z}(t, h)\big) > \delta(t, h) - 2 - n^{- 2 / 3} \tau_{\alpha}\big).
	\end{split}
	\end{align}
	
	Assume $P_1(t, h)$ and $P_2(t, h)$ are the true link probability matrices of groups $\{A_i\}_{i = t - h + 1}^{t}$ and $\{A_i\}_{i = t + 1}^{t + h}$. For proof convenience later, we denote matrices $B_1(t,h)$, $B_2(t,h)$, $V_2(t,h)$ all with zero diagonals and for all $i \neq j$,	
	$$B_{1, i j}(t,h) = \frac{P_{1, i j}(t, h) - P_{2, i j}(t, h)}{\sqrt{ (n - 1)\bigg\{\frac{1}{h} \hat{P}_{1, i j}(t, h) \left[ 1 - \hat{P}_{1, i j}(t, h)\right] + \frac{1}{h} \hat{P}_{2, i j}(t, h) \left[ 1 - \hat{P}_{2, i j}(t, h)\right]\bigg\} }},$$
	$$B_{2, i j}(t,h) = \frac{[P_{1, i j}(t, h) - P_{2, i j}(t, h)] - [\bar{A}_{1, i j}(t, h) - \bar{A}_{2, i j}(t, h)]}{\sqrt{ (n - 1)\bigg\{\frac{1}{h} \hat{P}_{1, i j}(t, h) \left[ 1 - \hat{P}_{1, i j}(t, h)\right] + \frac{1}{h} \hat{P}_{2, i j}(t, h) \left[ 1 - \hat{P}_{2, i j}(t, h)\right]\bigg\} }},$$
	$$V_{2, i j}(t,h) = \frac{[ P_{1, i, j}(t, h) - P_{2, i j}(t, h)] - [\bar{A}_{1, i j}(t, h) - \bar{A}_{2, i j}(t, h)]}{\sqrt{ (n - 1)\bigg\{\frac{1}{h} P_{1, i j}(t, h) \left[ 1 - P_{1, i j}(t, h)\right] + \frac{1}{h} P_{2, i j}(t, h) \left[1 - P_{2, i j}(t, h)\right]\bigg\} }}.$$
	Then the lower bound of $\sigma_1(t,h)$ can be obtained:
	\begin{align*}
	\sigma_1\big(\hat{Z}(t, h)\big) & \geq \sigma_1( B_1(t,h)) - \sigma_1( B_2(t,h)) \\
	& = \left[ \sigma_1( V_1(t,h)) - \sigma_1( V_2(t,h))\right] [ 1 + o_p(n^{- 2 / 3})] \\
	& \geq \delta(t,h) - 2 - n^{- 2 / 3} \tau_{\alpha},
	\end{align*}
	with probability at most $1 - 2 \alpha + o_n(1)$. The last inequality follows by noting that $V_2(t,h)$ is a generalized Wigner matrix. Similarly with proof of Lemma \ref{Asymptotic distributions of Z_star for two-sample}, we have $P\big(n^{2 / 3} \big[\sigma_1( V_2(t,h)) - 2\big] \leq \tau_{\alpha}\big) \geq 1 - 2 \alpha + o_n(1)$. Combining this with \eqref{a}, we have
	\begin{equation*}
	P( T_{TW_1}(t, h) > \vartriangle_{T_{TW_1}}) = 1 - 2 \alpha + o_n(1) = (1 - 1 / n)^{1 / {(2 h)}} + o_n(1) \rightarrow 1.
	\end{equation*}
	
	The above result implies that with probability of 1, all and only the change-points will be selected at the thresholding steps. Therefore, we have
	$$\lim_{n \rightarrow \infty} P\big(\cJ = \hat{\cJ} \big) = 1.$$
	

	\section{Academic calendar of MIT 2004--2005}
	The academic calendar of MIT we use in this paper is illustrated as follows. 

	\begin{table}[htbp]
		\centering
		\caption{Academic calendar of MIT from July 20, 2004 to June 14, 2005.}
		{\scalebox{0.75}{\begin{tabular}{l p{15cm}}
					\toprule			    
					Date & Event \\
					\midrule				
					August 6, 2004 & \textbf{Deadline for doctoral students to submit application for Fall Term Non-Resident status; Thesis due for September degree candidates}.  \\ 
					August 12, 2004 & \textbf{Continuing students final deadline to pre-reg on-line}.\\	
					August 13, 2004 & \textbf{Last day to go off the September degree list}.  \\
					August 16--17, 2004 & Summer Session Final Exam Period. \\
					August 23, 2004 & \textbf{Grades due}. \\
					August 27, 2004 & Term Summaries of Summer Session Grades.\\
					August 30, 2004 & Graduate Student Orientation activities begin. \\
					August 31, 2004 & English Evaluation Test for International students. \\
					September 6, 2004 & Labor Day--Holiday. \\
					September 7, 2004 & \textbf{Registration day}. \\
					September 8, 2004 & \textbf{First day of classes}. \\
					September 9--17, 2004 & Physical Education Petition Period. \\
					September 10, 2004 & \textbf{Degree application deadline}. \\
					September 14, 2004 & Committee on Graduate School Policy Meeting. \\
					September 15, 2004 & Faculty officers recommend degrees to Corporation. \\
					September 24, 2004 & \textbf{Minor completion date}.  \\
					September 30, 2004 & Last day to sing up family health insurance or waive individual coverage. \\
					October 1, 2004 & Deadline for completing Harvard cross-registration. \\
					October 8, 2004 & \textbf{Last day to add subjects to Registration}.  \\
					October 11, 2004 & Columbus Day--Holiday. \\
					October 15--17, 2004 & Family Weekend.\\
					October 2, 20046 & Second quarter Physical Education classes begin.  \\
					November 1, 2004 & Half-term subjects offered in second half of term begin. \\
					November 11, 2004 & Veteran's Day--Holiday. \\
					November 17, 2004 & \textbf{Last day to cancel subjects from Registration}.\\
					November 25--26, 2004 & Thanksgiving Vacation--Holiday. \\
					December 1, 2004 & \textbf{On-line pre-registration for Spring Term begins}. \\
					December 3, 2004 & \textbf{Subjects with no final/final exam}. \\	
					December 9, 2004 & \textbf{Last day of classes}. \\									
					\bottomrule
		\end{tabular}}} \label{tab calendar 1}
	\end{table}

	\begin{table}[htbp]
		\centering
		{\scalebox{0.75}{\begin{tabular}{l p{15cm}}
					\toprule			    
					Date & Event \\
					\midrule
					December 10, 2004 & Last day to submit or change Advanced Degree Thesis Title. \\
					December 13--17, 2004 & \textbf{Final exam period}. \\
					December 14--22, 2004 & \textbf{Grade deadline}. \\
					December 18, 2004 & Winter Vacation begins--Holiday.\\
					December 30, 2004 & \textbf{Spring pre-registration deadline}. \\
					January 2, 2005 & Winter Vacation ends. \\
					January 3, 2005 & \textbf{Deadline for doctoral students to submit applications for Spring Term Non-Resident status}.  \\
					January 6, 2005 & Term Summaries of Fall Term Grades. \\
					January 7, 2005 & \textbf{Thesis due}.   \\
					January 10, 2005 & Second-Year and Third-Year Grades Meeting.  \\
					January 11, 2005 & Fourth-Year Grades meeting; Committee on Graduate School Policy Meeting.   \\
					January 13, 2005 & \textbf{Final deadline for continuing students to pre-reg on-line}. \\
					January 14, 2005 & \textbf{Thesis due}. \\
					January 17, 2005 & Martin Luther King, Jr. Day--Holiday.\\
					January 19--20, 2005 & C.A.P. deferred action meeting.\\
					January 26, 2005 & English Evaluation Test for International students. \\
					January 26--28, 2005 & Some advanced standing exams and postponed finals. \\
					January 28, 2005 & Last day of January Independent Activities Period. \\
					January 31, 2005 & \textbf{Registration day}.  \\	
					February 1, 2005 & \textbf{First day of classes}.  \\
					February 2--11, 2005 & Physical Education Petition Period. \\	
					February 3, 2005 & \textbf{Grades due}. \\
					February 4, 2005 & \textbf{Registration deadline}.\\
					February 7, 2005 & Term summaries of Grades for IAP. \\ 
					February 8, 2005 & Committee on Graduate School Policy Meeting. \\
					February 11, 2005 & C.A.P. February Degree Candidates Meeting. \\
					February 16, 2005 & Faculty Officers recommend degrees to Corporation. \\
					February 18, 2005 & \textbf{Minor completion date}. \\
					February 21, 2005 & Presidents Day--Holiday. \\
					February 22, 2005 & \textbf{Monday schedule of classes to be held}.\\
					February 28, 2005 & Last day to sing up for family health insurance or waive individual coverage.\\
					March 4, 2005 & \textbf{Last day to add subjects to Registration}.  \\
					March 21--25, 2005 & Spring Vacation--Holiday. \\
					March 28, 2005 & Half-term subjects offered in second half of term begin. \\
					March 30, 2005 & Fourth quarter Physical Education classes begin. \\
					April 1, 2005 & Last day to submit or change Advanced Degree Thesis Title.\\
					April 7--10, 2005 & Campus Preview Weekend. \\
					April 18--19, 2005 & Patriots Day--Holiday.   \\
					April 21, 2005 & \textbf{Last day to cancel subjects from Registration}. \\
					April 29, 2005 & \textbf{Thesis due}. \\
					May 2, 2005 & \textbf{On-line pre-registration for Fall Term and Summer Session begins}.\\
					May 6, 2005 & \textbf{Subjects with no final/final exam}.  \\
					May 12, 2005 & \textbf{Last day of classes}. \\
					May 16--20, 2005 & \textbf{Final exam week}. \\
					May 17--24, 2005 & \textbf{Grade deadline}.  \\ 
					May 20, 2005 & \textbf{Last day to go off the June degree list}. \\
					May 26, 2005 & \textbf{Department grades meetings}.\\
					May 27, 2005 & Fourth-Year Grades Meeting.\\
					May 30, 2005 & Memorial Day--Holiday. \\
					May 31, 2005 & \textbf{Fall pre-registration deadline}. \\
					June 1, 2005 & First-Year Grades Meeting. \\
					June 2, 2005 & Doctoral Hooding Ceremony. \\
					June 3, 2005 & \textbf{Commencement}. \\
					June 14, 2005 & C.A.P. deferred action meeting. \\
					\bottomrule
		\end{tabular}}} 
	\end{table}

	\bibliographystyle{apalike}
	\bibliography{ref-hypothesis.bib}

\end{document}